\newif\ifnotes
\newif\ifcr
\notestrue
\crfalse

\ifnotes
\newcommand{\omri}[1]{$\ll$\textsf{\color{blue} Omri: { #1}}$\gg$}

\else
\newcommand{\omri}[1]{}
\fi

\documentclass[a4paper,11pt]{article}
\usepackage{footmisc} % Needs to be before the "hyperref" usage, as it sometimes can cause problems with footnotes.
\usepackage[hypertexnames=false,colorlinks=true,citecolor=Maroon]{hyperref}
\usepackage{newtxtext}
\usepackage{amsthm}
\usepackage{fullpage}
\usepackage{amsmath,amssymb,enumerate,algorithm,algorithmic,latexsym}
\usepackage{tikz}
\usetikzlibrary{fpu}
\usepackage{breakcites}
\usepackage{float}
\newfloat{algorithm}{H}{lop}
\usepackage{color}
\usepackage{colortbl}
\definecolor{Maroon}{cmyk}{0, 0.87, 0.68, 0.32}
\usepackage{url}
\usepackage{graphicx}
\usepackage{xspace}
\usepackage{subfigure}
\usepackage{amsfonts}
\usepackage{caption}
\usepackage{enumitem}
\usepackage{soul}
\usepackage[normalem]{ulem}
\usepackage{url}
\usepackage{graphicx}
\usepackage{bookmark}
\numberwithin{algorithm}{section}

\renewcommand{\paragraph}[1]{\vspace{1.5mm}\noindent \textbf{#1}}

\newcommand{\poly}{\mathrm{poly}}
\newcommand{\sk}{\mathsf{sk}}
\newcommand{\pk}{\mathsf{pk}}

\newcommand{\prf}{\mathsf{PRF}}
\newcommand{\prfGen}{\prf.\mathsf{Gen}}
\newcommand{\prfF}{\prf.\mathsf{F}}
\newcommand{\prfk}{\mathsf{prfk}}

\newcommand{\fhe}{\mathsf{FHE}}
\newcommand{\fheGen}{\mathsf{FHE.Gen}}
\newcommand{\fheEnc}{\mathsf{FHE.Enc}}
\newcommand{\fheDec}{\mathsf{FHE.Dec}}
\newcommand{\fheEval}{\mathsf{FHE.Eval}}
\newcommand{\ciph}{\mathsf{ct}}

\newcommand{\Cir}{{C}}
\newcommand{\fheSim}{\mathsf{Sim}}
\newcommand{\fheExt}{\mathsf{Ext}}
\newcommand{\fhek}{\mathsf{fhek}}
\newcommand{\evciph}{\hat{\ciph}}

\newcommand{\sfeEval}{\mathsf{SFE.Eval}}

\newcommand{\prot}[2]{\ve{#1,#2}}
\newcommand{\view}{\mathsf{OUT}}
\newcommand{\zkmD}{\mathsf{D}}
\newcommand{\A}{\mathsf{A}^*}
\newcommand{\zkP}{\mathsf{P}}
\newcommand{\zkV}{\mathsf{V}}
\newcommand{\zkSet}{\mathsf{Setup}}
\newcommand{\zkVSet}{\mathsf{VSetup}}
\newcommand{\zkSim}{\mathsf{Sim}}
\newcommand{\zkmP}{\zkP^*}

\newcommand{\sigmaP}{\Sigma.\mathsf{\zkP}}
\newcommand{\sigmaV}{\Sigma.\mathsf{\zkV}}

\newcommand{\qsigmaP}{\Xi.\zkP}
\newcommand{\qsigmaV}{\Xi.\zkV}
\newcommand{\qsigmaS}{\Xi.\zkSim}
\newcommand{\pke}{\mathsf{PKE}}
\newcommand{\pkeGen}{\pke.\mathsf{Gen}}
\newcommand{\pkeEnc}{\pke.\mathsf{Enc}}
\newcommand{\pkeDec}{\pke.\mathsf{Dec}}

\newcommand{\Com}{\mathsf{Com}}

\newcommand{\Disting}{\mathsf{D}^*}
\newcommand{\Hyb}{\mathsf{Hyb}}

\newcommand{\ket}[1]{|{#1}\rangle}

\newcommand{\TD}{\mathrm{TD}}

\newcommand{\lang}{\mathcal{L}}
\newcommand{\langyes}{\mathcal{L}_{yes}}
\newcommand{\langno}{\mathcal{L}_{no}}

\newcommand{\rel}{\mathcal{R}}
\newcommand{\ins}{x}

\newcommand{\wit}{w}
\newcommand{\qwit}{\ket{w}}

\newcommand{\real}{\mathsf{Real}}
\newcommand{\simulated}{\mathsf{Simulated}}
\newcommand{\crs}{\mathsf{crs}}
\newcommand{\zkcrs}{\mathsf{\tilde{crs}}}
\newcommand{\td}{\mathsf{td}}
\newcommand{\nizk}{\mathsf{NIZK}}
\newcommand{\nizkP}{\mathsf{NIZK.\zkP}}
\newcommand{\nizkV}{\mathsf{NIZK.\zkV}}
\newcommand{\nizkSet}{\mathsf{NIZK.\zkSet}}

\newcommand{\nizkSim}{\mathsf{NIZK.\zkSim}}
\newcommand{\pvk}{\mathsf{pvk}}
\newcommand{\svk}{\mathsf{svk}}
\newcommand{\mpvk}{\mathsf{pvk^*}}
\newcommand{\ek}{\mathsf{ek}}

\newcommand{\textabbrevstyle}[1]{\mbox{#1}}
\newcommand{\textabbrevstylebol}[1]{\mbox{\textbf{#1}}}
\newcommand{\newtextabbrev}[1]{\expandafter\newcommand\csname #1\endcsname{\textabbrevstyle{#1}\xspace}}
\newcommand{\newtextabbrevbol}[1]{\expandafter\newcommand\csname #1\endcsname{\textabbrevstylebol{#1}\xspace}}
\newcommand{\renewtextabbrevbol}[1]{\expandafter\renewcommand\csname
#1\endcsname{\textabbrevstylebol{#1}\xspace}}

\newtextabbrevbol{EXP}
\newtextabbrevbol{Dtime}
\newtextabbrev{PRG}
\newtextabbrev{PRF}
\newtextabbrev{OWF}
\newtextabbrev{PPT}
\newtextabbrev{QPT}
\newtextabbrev{EOWF}
\newtextabbrev{GEOWF}
\newtextabbrev{ZAP}
\newtextabbrevbol{Ntime}
\newtextabbrevbol{coAM}
\newtextabbrevbol{NP}
\newtextabbrevbol{MIP}
\newtextabbrevbol{NEXP}
\newtextabbrevbol{SZK}
\newtextabbrevbol{BPP}
\newtextabbrevbol{coNP}
\newtextabbrevbol{coMA}
\newtextabbrevbol{TFNP}
\newtextabbrev{NIWI}
\newtextabbrevbol{AM}
\newtextabbrev{SNARG}
\newtextabbrev{SNARK}
\newtextabbrev{coRP}
\newtextabbrev{NIUA}
\newtextabbrev{UA}
\newtextabbrev{ZK}
\newtextabbrev{WZK}
\newtextabbrev{WH}
\newtextabbrev{AI}
\newtextabbrev{ECRH}
\newtextabbrev{PIR}
\newtextabbrev{WI}
\newtextabbrev{WIPOK}
\newtextabbrev{CDS}
\newtextabbrev{POK}
\newtextabbrev{FE}
\newtextabbrev{IO}
\newtextabbrev{XIO}
\newtextabbrev{SXIO}
\newtextabbrev{SKFE}
\newtextabbrev{PKFE}
\newtextabbrev{PPRF}
\newtextabbrev{LWE}
\newtextabbrev{QLWE}
\newtextabbrev{PKE}

\newcommand{\QMA}{\textabbrevstylebol{QMA}}

\newtheorem{definition}{Definition}[section]

\newtheorem{corollary}{Corollary}[section]
\newtheorem{theorem}{Theorem}[section]
\newtheorem{claim}{Claim}[section]

\newtheorem{proposition}{Proposition}[section]

\theoremstyle{remark}

\newcommand{\figref}[1]{Figure~\protect\ref{#1}}

\newcommand{\proref}[1]{Protocol~\protect\ref{#1}}

\newenvironment{boxfig}[2]{\begin{figure}[#1]\fbox{\begin{minipage}{\linewidth}
                        \vspace{0.2em}
                        \makebox[0.025\linewidth]{}
                        \begin{minipage}{0.95\linewidth}
            {{
                        #2 }}
                        \end{minipage}
                        \vspace{0.2em}
                        \end{minipage}}}{\end{figure}}

\newcommand{\pprotocol}[4]{
\begin{boxfig}{h!}{
\begin{center}
\textbf{#1}
\end{center}
    #4
\vspace{0.2em} } \caption{\label{#3} #2}
\end{boxfig}
}

\newcommand{\protocol}[4]{
\pprotocol{#1}{#2}{#3}{#4} }
\newcommand{\negl}{\mathrm{negl}}

\newcommand{\ve}[1]{\langle #1 \rangle}

\newcommand{\set}[1]{\left\{#1\right\}}
\newcommand{\abs}[1]{\left|#1\right|}

\newcommand{\pST}{\; \middle\vert \;}

\newcommand{\zo}{\{0,1\}}

\newcommand{\Nat}{\mathbb{N}}
\newcommand{\secp}{\lambda}

\newcommand{\cupdot}{\mathbin{\mathaccent\cdot\cup}}

\title{Multi-theorem (Malicious) Designated-Verifier NIZK for QMA}

\author{Omri Shmueli\thanks{Tel Aviv University, \texttt{omrishmueli@mail.tau.ac.il}. Supported by ISF grants 18/484 and 19/2137, by Len Blavatnik and the Blavatnik Family Foundation, and by the European Union Horizon 2020 Research and Innovation Program via ERC Project REACT (Grant 756482).}}
\date{\vspace{-5ex}}
\date{}

\begin{document}
\maketitle

\begin{abstract}
	We present the first non-interactive zero-knowledge argument system for QMA with multi-theorem security.
	Our protocol setup constitutes an additional improvement and is constructed in the malicious designated-verifier (MDV-NIZK) model (Quach, Rothblum, and Wichs, EUROCRYPT 2019), where the setup consists of a trusted part that includes only a common uniformly random string and an untrusted part of classical public and secret verification keys, which even if sampled maliciously by the \emph{verifier}, the zero knowledge property still holds.
	The security of our protocol is established under the Learning with Errors Assumption.
	
	Our main technical contribution is showing a general transformation that compiles any sigma protocol into a reusable MDV-NIZK protocol, using NIZK for NP. 
	Our technique is classical but works for quantum protocols and allows the construction of a reusable MDV-NIZK for QMA. 
\end{abstract}

\thispagestyle{empty}
\newpage
\tableofcontents
\thispagestyle{empty}
\newpage
\pagenumbering{arabic}

\section{Introduction} \label{sec:intro}
Zero-knowledge protocols allow to prove statements without revealing anything but the mere fact that they are true. Since their introduction by Goldwasser, Micali, and Rackoff \cite{GoldwasserMR89} they have had a profound impact on modern cryptography and theoretical computer science at large.
While standard zero-knowledge protocols are interactive, Blum, Feldman, and Micali \cite{blum2019non} introduced the concept of a non-interactive zero-knowledge (NIZK) protocol, which consists of a single message sent by the prover to the verifier. NIZK protocols cannot exist in the plain model (i.e. a language with such a NIZK protocol can be decided by an efficient algorithm) but can be realized with a pre-computed setup. The point of the setup is that it can be computed instance-independently and usually, the setup is executed by a trusted third party that generates and publishes a string of bits and sometimes trapdoors are handed to the prover or verifier (or both).

Although existing zero-knowledge protocols for NP cover an array of diverse tasks and in particular, under standard computational assumptions it is known how to construct NIZK protocols for NP \cite{canetti2019fiat, peikert2019noninteractive, brakerski2020nizk}, far less is known about the class QMA, the quantum generalization of NP.
This knowledge gap between NP and QMA, which is present in both interactive and non-interactive zero-knowledge protocols, stems from the fact that many of the techniques that work for NP and more precisely, classical information-processing, usually fail when are needed for the processing of quantum information.

The first expression of the gap between classical and quantum NIZK protocols is that of setup requirements, that is, how much trust and resources the setup needs.
For example, the standard setup in NIZK is called the common reference string (CRS) model, where the trusted party samples a classical string from some specified distribution and publishes it (no trapdoors are handed to either prover or verifier in this model). If the reference string is simply uniformly random then the setup is in the common \emph{random} string model, which is considered to require minimal trust in the NIZK setting, as the trusted party holds no trapdoors whatsoever. NIZK arguments for NP are known to exist in the common random string model under LWE \cite{canetti2019fiat, peikert2019noninteractive}.
In current QMA constructions the setup is comprised at least of a common \emph{reference} string sampled by the trusted party, and an additional public and secret verification keys $(\pvk, \svk)$ where $\pvk$ is published along with the CRS and $\svk$ is kept by the verifier, such that either:
\begin{itemize}
	\item
	$\pvk$ is a quantum state that needs to stay coherent while waiting for the proof by the prover, or
	
	\item
	The pair $(\pvk, \svk)$ can be sampled only by the trusted party and not the verifier.
\end{itemize}

Aside from the above, perhaps the most basic missing part between NIZK protocols for NP and QMA is the existence (or inexistence) of multi-theorem security.
Multi-theorem security considers the \emph{reusability of the setup}, that is, once the setup is computed, any prover can send a proof by a single message repeatedly for many different statements and there is no need to re-compute the setup for every new proof sent and in relation to the above QMA setups: once the CRS and public verification key are published, they are reusable.
However, although multi-theorem security provides the main efficiency advantage to a NIZK protocol over an interactive protocol, we currently don't have non-interactive zero-knowledge protocols for QMA with reusable setups.

Given the gap of knowledge in NIZK techniques between NP and QMA, improving the power of NIZKs for QMA seem as a natural cryptographic goal which we explore in this work.

\subsection{Results}\label{subsec:results}
Under the Learning with Errors (LWE) assumption \cite{Regev09} we resolve the above open question. 
Specifically, we construct a NIZK argument for QMA with multi-theorem security and reduce setup requirements by proving security in the following model: 
\begin{enumerate}
	\item
	The trusted party samples only a common random string $\crs$.
	
	\item
	Given $\crs$, any verifier can sample a pair of classical public and secret verification keys $(\pvk, \svk)$, in particular it is possible that the published $\pvk$ is maliciously-generated.
\end{enumerate}
Given $\crs$ and $\pvk$, any prover can repeatedly give a non-interactive zero-knowledge proof by a single quantum message $\ket{\pi}$.
The above setup model is introduced by Quach, Rothblum, and Wichs in \cite{quach2019reusable} as the malicious designated-verifier model (MDV-NIZK), and has the same minimal trust requirements as the common random string model (but is privately verifiable).

\begin{theorem} [informal] \label{intro_thm}
	Assuming that LWE is hard for polynomial-time quantum algorithms, there exists a reusable, non-interactive computational zero-knowledge argument system for QMA in the malicious designated-verifier model.
\end{theorem}

\paragraph{Main Technical Contribution: General Sigma Protocol MDV-NIZK Compilation.}
Technically, we deviate completely from previous NIZK constructions for QMA and our main contribution is showing how given a NIZK for NP it is possible to compile \emph{any} general sigma protocol into a reusable MDV-NIZK protocol.
Our technique is simple and purely classical but works also for quantum zero-knowledge protocols and in particular can be used for showing a reusable MDV-NIZK for QMA. Further details are given in the technical overview below.

\subsection{Technical Overview}
We next describe our construction of a multi-theorem-secure MDV-NIZK protocol for QMA. For a discussion about the possibility of constructing a NIZK protocol for QMA in the CRS model see subsection \ref{subsec:QMA_CRS}, and for an overview of NIZK models and previous work on NIZK for QMA see subsection \ref{subsec:QMA_previous_work}.

We deviate from previous approaches of NIZK for QMA and take a different (and very natural) approach: Find a "classical anchor" in quantum zero-knowledge protocols and then solve the problem by having a NIZK for NP.
As such we currently restrict our attention to an even simpler, purely-classical question: Given any sigma protocol $(\sigmaP, \sigmaV)$, generically compile it into a multi-theorem-secure MDV-NIZK while assuming minimal properties of the protocol\footnote{In particular, we do not assume that the message $\alpha$ is classical.}. We will start with considering classical sigma protocols and later see what changes should take place in order for the technique to work for quantum protocols.

\paragraph{From a Sigma Protocol to a Single-Theorem-Secure MDV-NIZK.}
A sigma protocol is a 3-message public-coin proof system (with some mild zero knowledge properties), where the 3 messages are denoted by $\alpha$, $\beta$ and $\gamma$ (i.e. $\beta$ is a random string and is called "the challenge string").
Our first step is to construct a MDV-NIZK protocol with only single-theorem security out of a sigma protocol and is very simple.

In a sigma protocol, since the verifier's message $\beta$ is a random string it is independent of any other information, additionally, our second need from it is that it stays hidden (until after the prover sends its first message $\alpha$). The verifier can compute its public verification key, which is computed instance-independently, as a function of $\beta$: The public verification key $\pvk$ is an FHE-encrypted random challenge $\beta$ and the secret verification key $\svk$ is the FHE decryption key and the challenge string,
$$
\pvk = \fheEnc_{\fhek}(\beta), \enspace \svk = (\beta, \fhek) \enspace . 
$$
Given the public verification key $\pvk$, the 1-message proof procedure for $x \in \lang$ goes as follows:
\begin{itemize}
	\item
	$\zkP$ computes the first sigma protocol message $\alpha \gets \sigmaP(x, \wit)$, where $\wit \in \rel_{\lang}(x)$.
	
	\item
	$\zkP$ computes $\gamma$ the last protocol message under the encryption, that is, $\zkP$ performs the homomorphic evaluation $\evciph_\zkP \gets \fheEval(\sigmaP_3, \fheEnc_{\fhek}(\beta))$.
	
	\item
	As the proof, $\zkP$ sends $\alpha$ out in the open and $\gamma$ under the encryption, that is, the proof is $\pi = (\alpha, \evciph_{\zkP})$.
\end{itemize}
In order for the proof to stay zero-knowledge, the homomorphic evaluation needs to be circuit-private.
The verification algorithm is straightforward: Given $\svk$, an instance $x$ and a proof $\pi = (\alpha, \evciph_{\zkP})$, the verifier decrypts $\evciph_{\zkP}$ to get $\gamma$, and accepts iff the sigma protocol verifier accepts $\sigmaV(x, \alpha, \beta, \gamma) = 1$.

\paragraph{Is the Above Protocol Multi-Theorem-Secure?}
While it is intuitively clear that the described construction is secure for a single use of the setup (that is, the above should, with some modifications, yield a single-theorem-secure MDV-NIZK) it is provably not multi-theorem-secure by a standard attack.
Sigma protocols are usually parallel repetitions of 3-message zero-knowledge protocols, for example, consider the sigma protocol which is the parallel repetition of the zero-knowledge protocol for Graph Hamiltonicity \cite{blum1986prove}, which is as follows: Given a Hamiltonian cycle $C$ in a graph $G = (V, E)$, the prover samples a random permutation $\varphi : V \rightarrow V$ of the vertices and commits to the permuted graph $\varphi(G)$\footnote{That is, the prover commits to all of the cells in the adjacency matrix that represents the graph $\varphi(G)$.}.
The verifier then sends a random bit $b$, and the prover answers accordingly:
\begin{itemize}
	\item
	If $b = 0$ it is considered as a validity check, and the prover opens all commitments and sends $\varphi$. The verifier accepts if indeed the committed graph is $\varphi(G)$.
	
	\item
	If $b = 1$ it is considered as the cycle check, and the prover opens commitments only for the subgraph $\varphi(C)$.
	The verifier accepts if the opening shows a Hamiltonian cycle.
\end{itemize} 

If the sigma protocol used in the above MDV-NIZK construction is the parallel repetition of the zero-knowledge protocol for Hamiltonicity\footnote{We take the Hamiltonicity protocol only as a concrete easy example and in fact any other sigma protocol can take the role of this protocol in our context of attacking the soundness.}, then there is a polynomial-time malicious prover $\zkmP$ that given multiple access to the verifier's verdict function $\zkV(\svk, \cdot)$ using the same public/secret verification key pair, can decode the encrypted challenge string $\beta$ (which is polynomially-many random bits, each bit is for the $i$-th parallel repetition of the zero-knowledge protocol) and consequently break the soundness.

$\zkmP$ takes a Hamiltonian graph $G$ and a Hamiltonian cycle $C$ in it, and will decode the entire $\beta = (b_1, b_2, \cdots, b_k)$ bit-by-bit: To decode $b_i$, $\zkmP$ will honestly execute the zero-knowledge protocol prover's algorithm for all indices but index $i$ (that is, for all $j \neq i$, it will honestly compute $\Com(\varphi_j(G))$ and under the encryption, the opening of either the entire graph and the permutation of just the cycle $\varphi_j(C)$), for which it is going to operate as follows.
$\zkmP$ will guess that $b_i = 0$ and send a commitment to a permutation of the graph out in the open and under the encryption act as if $b_i = 0$ regardless of the actual value of $b_i$.
By the verifier's acceptance or rejection it will know whether the bit was $0$ or $1$.
After decoding $\beta$ the prover can now use this information to "prove" that any graph $G$ is Hamiltonian.

\paragraph{From Single-Theorem to Multi-Theorem Security.}
In the above attack the prover heavily relied on a specific operation: It uses a yes-instance (in the above case, a Hamiltonian graph $G$), in order to decode the random challenge $\beta$ and then goes on to use the knowledge of $\beta$ to give a false proof for a no-instance (again, in the above, a non Hamiltonian graph $G^*$).

Crucially, $\zkmP$ does not know how to decode $\beta$ when the graph is not Hamiltonian.
More specifically, in the above we decode $\beta$ bit-by-bit rather than all at once, and this ability comes from the fact that $G$ is Hamiltonian and the zero-knowledge protocol is complete, thus $\zkmP$ can be sure that if it honestly executes the zero-knowledge protocol for all indices but $i$, the only index that can make the proof get rejected is $i$. In this isolation, checking whether the challenge bit $b_i$ is $0$ or $1$ becomes trivial.
However, if the graph is not Hamiltonian then the prover cannot know which index made the proof get rejected because all $k$ indices are prone to rejection. Formally, by the soundness of the sigma protocol, we know that the answer from the verdict function of the verifier in this case will always be a rejection for any polynomial (or even sub-exponential) number of queries, with overwhelming probability. This means in particular that the prover cannot decode anything through the oracle access to the verdict function.

Our fix to the first protocol is based on the above observation: If we could make the random challenge $\beta$ change with the instance at hand it seems that the decoding attack is neutralized, because even if the prover decodes $\beta_G$ the challenge for a Hamiltonian graph $G$, it doesn't have information about $\beta_{G^*}$ the challenge of some non Hamiltonian $G^*$.
Since the instance \emph{$x$ is in particular a classical string} we can make the challenge change with the instance:
The public verification key will not be an encrypted challenge $\beta$ but instead will be a secret key $\prfk$ of a pseudorandom function $\prf$. The prover will compute $\alpha$ out in the open as before but the homomorphic evaluation changes: under the encryption, $\zkP$ will compute the challenge string as the PRF's output on the instance $\beta_x = \prfF_{\prfk}(x)$, and then compute $\gamma$ for the challenge $\beta_x$.

\paragraph{Extraction by Non-interactive Zero Knowledge for NP.}
Up to this point we only came close to constructing a provably-secure MDV-NIZK. Indeed, we didn't even use any NIZK tools yet for NP, and in order to prove the security of our construction we need knowledge extraction from both the prover and verifier.

To prove soundness, our thought process is roughly the following: We know that the prover computes $\gamma$ obliviously under the FHE, more precisely, it homomorphically evaluates the circuit $C_{x, r}$ that computes $\beta_x = \prfF_{\prfk}(x)$ and then given $\beta_x$ computes $\gamma$. The part of the circuit $C_{x, r}$ that computes $\gamma$ from $\beta_x$ is the "non-trivial" part of the circuit and is determined by a secret string $r$ (which is the information that the honest sigma protocol prover uses in order to compute $\gamma$, this information is the randomness of the prover and possibly the witness).
If we could extract $r$ from a prover (e.g. by the prover giving a proof of knowledge on the non-trivial part of the circuit $C_{x, r}$) that successfully cheats in the NIZK protocol then we could get a successfully cheating prover for the sigma protocol and thus prove security. To see this, note that by the hiding of the FHE and by the pseudorandomness of the PRF, even if as the public verification key we send an encryption of $0$ instead of an encryption of the PRF secret key, the string $r$ still needs to yield a circuit $C_{x, r}$ that does well in generating a satisfying $\gamma$ for a now-truly-random challenge $\beta$.

On the zero knowledge side we also need extraction, and we start with recalling a basic property of a sigma protocol: if we know the challenge string $\beta$ before sending the first message $\alpha$ then we can simulate a view that is indistinguishable from the real interaction with the honest prover.
This means that the information we want to extract from the malicious verifier is the secret PRF key $\prfk$ that in particular holds the information for obtaining $\beta_x$.

We solve both extraction tasks by a combination of a two-sided NP NIZK and a public-key encryption scheme with pseudorandom public keys.
Given the existence of a PKE scheme $(\pkeGen, \pkeEnc, \pkeDec)$ with pseudorandom public keys of length $\ell$ we take the common random string of our protocol to be (1) the common random string of an NP NIZK $(\nizkSet, \nizkP, \nizkV)$ protocol which we denote with $\crs$, concatenated with (2) a random string of length $\ell$ which we denote with $\ek$ (for extraction key).

We will let each of the parties encrypt, using $\pkeEnc_{\ek}(\cdot)$, the secrets that we want to extract and then use the NIZK to prove consistency between the content of the PKE encryption and the protocol computations. More precisely, as part of its 1-message proof, the prover will give a proof $\pi_{\zkP}$ that the string $r$ encrypted using the PKE yields the (canonical) circuit $C_{x, r}$ that it used for the (circuit-private) homomorphic evaluation that generated $\gamma$, and the verifier, as part of its public verification key, will give a proof $\pi_{\zkV}$ that the PRF key $\prfk$ that is encrypted using the PKE is the same key encrypted with the FHE.
Note that the information that the parties encrypt using a random string instead of a real PKE key stays secure due to the fact that a real key is indistinguishable from a random string, and thus an adversary that manages to break the PKE when it uses a random string as the public key can break the pseudorandomness property of the public keys.

When wanting to extract information (either in the soundness reduction or in the zero-knowledge simulation), we will sample $\ek$ using the PKE key-generation algorithm $(\ek, \sk) \gets \pkeGen$, and since the public keys are pseudorandom the change in key distribution won't be felt by either of the parties.
At that point the parties encrypt their secrets and prove they do so using the NIZK, and the extractor can just use the PKE decryption $\pkeDec_{\sk}(\cdot)$ to obtain the secrets.

\paragraph{Compiling Quantum Protocols.}
Our technique so far is entirely classical and compiles classical sigma protocols. We now ask whether it works to compile quantum sigma protocols. This can be answered in turn by answering the following question: what properties of the sigma protocol \emph{exactly} did we use in order for the MDV-NIZK protocol to work?

It can be verified that even if we don't assume nothing on the sigma protocol that we compile, every action in the MDV-NIZK protocol except the homomorphic evaluation of the circuit $C_{x, r}$ can stay exactly the same.
Regarding the homomorphic evaluation, the issue that we have is the following: In order to still be able to extract the information $r$ of the circuit $C_{x, r}$ from the prover, the computation that takes $\beta_x$ and outputs $\gamma$ needs to be a classical circuit.
This is not necessarily the case in a quantum protocol. For example, in the quantum zero-knowledge protocol for QMA of \cite{broadbent2016zero} (which is also the basis for the quantum NIZK protocol of \cite{coladangelo2019non}), in order to generate $\gamma$ given $\alpha, \beta$, first a quantum Clifford operation that is chosen with respect to $\beta$ needs to be executed on $\alpha$, followed by a measurement. Then, the prover proves in ZK that the classical string obtained by the measurement satisfies some properties\footnote{in that protocol it is also needed that the verifier itself makes the Clifford operation and measurement, which makes the protocol more challenging to use for a NIZK protocol.}.
Luckily, we identify a different quantum protocol that in fact does satisfy the property that $\gamma$ can be computed by an entirely classical circuit.

We consider the Consistency of Local Density Matrices (CLDM) problem \cite{liu2006consistency}, which is a QMA problem with some special properties. In \cite{broadbent2019zero} Broadbent and Grilo show that CLDM is QMA-complete and how to construct a very simple quantum zero-knowledge protocol for it.
The \cite{broadbent2019zero} zero-knowledge protocol for CLDM is as follows: Given a quantum witness $\qwit$, the protocol starts with the prover sending a quantum one-time pad encryption of $\qwit$ as the message $\alpha$. More precisely, for a length-$l$ witness it samples classical random pads $a, b \gets \{ 0, 1 \}^l$, applies
$$
\bigotimes_{i \in [l]} \left( X^{a_i}\cdot Z^{b_i} \right) \cdot \qwit \enspace ,
$$
and then sends as $\alpha$ the transformed quantum state and classical commitments to the QOTP keys $a, b$.
For a random challenge $\beta$, the prover response $\gamma$ is an opening to part of the state.
We find the CLDM problem and specifically the zero-knowledge protocol for it especially attractive for our purposes as \emph{$\gamma$ is only a function of the randomness of the prover and the challenge $\beta$}, which in particular means that the circuit $C_{x, r}$ can stay classical in our setting.

Finally, by using the sigma protocol yielded by the parallel repetition of the zero-knowledge protocol from \cite{broadbent2019zero} we obtain a clean and simple non-interactive computational zero-knowledge argument system for the class QMA in the malicious designated-verifier model:

\begin{enumerate}
	\item {\bf Common Random String:} $(\crs, \ek)$.
	
	\item {\bf Public and Secret Verification Keys:} $\prfk \gets \prfGen(1^\secp), \fhek \gets \fheGen(1^\secp)$,
	$$
	\pvk = \big(\fheEnc_{\fhek}(\prfk), \; \pkeEnc_{\ek}(\prfk), \;\pi_\zkV \big), \enspace \svk = \big( \prfk, \; \fhek \big) \enspace .
	$$
\end{enumerate}
For any prover that wishes to give a proof for an instance $x \in \langyes$, it executes the following:
\begin{itemize}
	\item {\bf Proof:}
	If $\pi_\zkV$ is valid, $\zkP$ computes $\alpha \gets \qsigmaP(\qwit ; r)$ and sends
	$$
	\ket{\pi} = \big( \alpha, \; \fheEval(C_{x, r}, \fheEnc(\prfk)), \; \pkeEnc_{\ek}(r), \; \pi_\zkP \big) \enspace .
	$$
\end{itemize}

\subsection{Related Work}
In this section we discuss the main challenges in the construction of non-interactive zero-knowledge protocols for QMA (specifically in the CRS model) and the previous works on QMA NIZKs.

\subsubsection{Can we Build a NIZK protocol for QMA in the CRS model?} \label{subsec:QMA_CRS}
In short, the answer to the above question is that we don't know, and this section does not aim to answer it. This section is intended to give some evidence to why constructing a NIZK for QMA in the CRS model seem to require a different set of techniques from what we currently have for NP.
In what follows we will start with briefly recalling how NIZKs for NP are constructed and then understand why current approaches fail in the setting of quantum proofs.

\paragraph{NP, Fiat-Shamir and Correlation Intractability.} \label{par:FS_and_NP}
In order to construct a non-interactive zero-knowledge protocol for NP under standard assumptions, the construction starts with a sigma protocol $(\sigmaP, \sigmaV)$.
To make the protocol non-interactive, the Fiat-Shamir transform is applied: By assuming public oracle access to a random function $F$, the prover applies it to $\alpha$ and treat its (random-string) output $F(\alpha)$ as the challenge string $\beta$. It then computes $\gamma$ and sends all of this information to the verifier, who makes sure that $\beta$ was rightfully generated $\beta = F(\alpha)$, and that the sigma protocol verifier $\sigmaV(\alpha, \beta, \gamma)$ accepts.
Since we don't know how to construct a cryptographic primitive that acts as a publicly-computable random function, the above protocol is secure only in the random oracle model, that is, only if we directly assume public access to such random function $F$.

In order to prove the security of the NIZK protocol in the standard model (with access to a common reference string rather than a random oracle), the final part of the construction involves swapping the random function $F$ with a new, special hash function $H$ - this general technique of swapping $F$ with a special hash function $H$ is usually called the Correlation Intractability (CI) paradigm \cite{canetti2004random}.
The properties of the hash function $H$ or the meaning of correlation intractability are less relevant to this overview, but it is suffices to say that under the LWE assumption it is known how to construct a hash function $H$ that can be swapped with $F$ in the FS transform and where the protocol can be proven secure \cite{canetti2019fiat, peikert2019noninteractive}.

\paragraph{Can we use Known Classical NIZK Techniques for Quantum Protocols?} \label{par:FS_and_QMA}
There are two known routes for getting a quantum-secure NIZK for NP in the CRS model, the first is through the FS transform and CI (which also uses only standard assumptions, described above) and the second is through the hidden bits model and indistinguishability obfuscation.
It is natural to ask whether we can use these techniques for QMA (the question of whether the FS transform can be used for quantum protocols was asked as one of the open questions in section 1.4 of \cite{broadbent2019zero}).

We first review the ability to use the FS transform (and in particular correlation intractability) for QMA and explain why there is an issue with the no-cloning theorem.
In the quantum setting, sigma protocols $(\qsigmaP, \qsigmaV)$ \cite{broadbent2019zero, broadbent2016zero} are quite the same but with the main difference that the first message $\alpha$ is quantum (and of course, the prover takes as input a quantum witness $\qwit$ rather than classical).
Recall that when we use the FS transform on a sigma protocol in order to generate a NIZK, for the protocol to be complete, when the parties act honestly then the verifier needs to verify that the random function $F$ yields the challenge, that is $F(\alpha) = \beta$. This means that now $F$ needs to be a quantum transformation such that for $x \in \langyes$ and an honestly generated $\alpha \gets \qsigmaP(\qwit)$, $F(\alpha)$ is always the same classical string (with overwhelming probability).
Now, denote by $s$ the classical string s.t. $F(\alpha) = s$, and we have a generating circuit for the quantum witness: $\qwit = \qsigmaP^{\dagger}(\cdot)\cdot F^{\dagger}\cdot \ket{s}$, where the inverse versions of $F$ and $\qsigmaP$ are purified.
This seems to violate the no-cloning theorem in the following manner: the prover gets a copy of the witness and can generate a generating circuit for the witness state, this circuit can be used to generate arbitrarily many copies of the state.
Finally, because we can always consider a trivial language with a dummy witness, and take the quantum witness to be some unclonable state (for example, a pseudorandom quantum state) we get a contradiction to the no-cloning theorem.

Even if we aim to construct a NIZK using the FS transform for QCMA, the subclass of QMA where the verification algorithm is still quantum but the witness is classical, the problem is not seemed to be solved.
The reason, is that we don't know how to construct sigma protocols for QCMA where the first message $\alpha$ is classical, and the same contradiction to the no-cloning theorem holds.

The second known route of obtaining a quantum-secure NIZK protocol for NP in the CRS model is through the hidden bits model \cite{feige1999multiple} which is implementable by sub-exponentially-secure indistinguishability obfuscation \cite{bitansky2016perfect}.
In the hidden bits model, intuitively (and roughly), the trusted party samples as the common reference string a commitment to a string sampled from some distribution (where by using a trapdoor permutation, the prover can open the commitments efficiently), and the prover proves that the instance at hand $x \in \langyes$ satisfies some property related to the string underlying the commitments.
Even if we are willing to assume the very strong cryptographic assumptions which are needed for the realization of this protocol (i.e. sub-exponentially-secure post-quantum indistinguishability obfuscation), it is currently unknown how to use the hidden bits model to instantiate non-interactive zero-knowledge quantum protocols.

\subsubsection{Relaxations of the CRS Model and Previous Work} \label{subsec:QMA_previous_work}
The constructions of NIZKs for NP discussed in subsection \ref{subsec:QMA_CRS} are implicitly in the CRS model, where the setup consists of a string that is sampled and published by the trusted party, in particular, nor the prover or verifier hold any trapdoors over the setup.
Sometimes when it is unknown how to build a NIZK in the CRS model (or unknown how to minimize the assumptions for building one) we turn to relaxations of the CRS model.
For example, in the designated-verifier model (DV-NIZK) \cite{pass2006construction} the trusted party samples, along with the CRS, a pair of public and secret verification keys $(\pvk, \svk)$, publishes $\pvk$ along with the CRS and hands $\svk$ only to the verifier. Another example is the designated prover model (DP-NIZK) \cite{kim2019multi}, which is analogous to the DV-NIZK model, only that the prover is the one who gets a secret, now-\emph{proof} key.

It is a well known fact in the design of NIZKs that when the verifier holds a secret verification key (e.g. in the DV-NIZK model) then multi-theorem zero knowledge can be achieved generically by the compiler of \cite{feige1999multiple}, but multi-theorem soundness becomes non-trivial. For example, it is possible (and is sometimes provably the case) that the prover can decode the verifier's secret key by having access multiple times to the verifier's verdict function, consequently breaking the soundness of the protocol.
Indeed, one example is that until the works of \cite{quach2019reusable, lombardi2019new}, based on \cite{pass2006construction} it was only known how to get single-theorem-secure DV-NIZK for NP, and another example is that this is the current situation with QMA constructions of NIZK protocols.

The QMA NIZK protocol of Broadbent and Grilo \cite{broadbent2019zero} is in the secret parameters model (i.e. the protocol is both designated-prover and designated-verifier and both parties get secret keys from the trusted party) but is a proof system and has \emph{statistical} soundness rather than the computational soundness we achieve. The protocol of Coladangelo, Vidick and Zhang \cite{coladangelo2019non} is in a model that is somewhat between the common reference string model and the DV-NIZK model, where the trusted party samples a common reference string and the verifier itself samples a pair $(\pvk, \svk)$ where $\pvk$ is a quantum state.
Both of the abovementioned protocols are not reusable.

Outside of the standard model, an additional construction by Alagic, Childs, Grilo and Hung \cite{alagic2019non} yields a QMA NIZK protocol in the quantum random oracle model (with additional setup in the secret parameters model) which is both reusable and classical-verifier.

There are two main issues with letting the trusted party sample secret keys for any of the parties: First, the trust requirements of the setup now increase as the party receiving the secret key should assume that the trusted party handles its secret information securely. The second issue is that of centralization of computational resources: for example, in the DV-NIZK model, the trusted party is now responsible for sampling a fresh pair $(\pvk, \svk)$ for every new verifier that wishes to use the protocol, which is very different from the CRS setting where it samples a string and from that point on can terminate.

The \emph{malicious} designated-verifier (MDV-NIZK) model \cite{quach2019reusable, lombardi2019new} seeks to solve the above two problems, which is also the model of our protocol.
In the MDV-NIZK model the trusted party only samples a common random string, and then, any verifier wishing to use the protocol can sample \emph{by itself} a pair of classical keys $(\pvk, \svk)$ and publish $\pvk$. The protocol then stays secure even if the public key $\pvk$ is maliciously-generated.

\subsubsection*{Acknowledgments}
We deeply thank Nir Bitansky and Zvika Brakerski for helpful discussions during the preparation of this work.

\section{Preliminaries}\label{sec:prel}
\noindent We rely on standard notions of classical Turing machines and Boolean circuits:

\begin{itemize}
	\item
	A \PPT algorithm is a probabilistic polynomial-time Turing machine.
	
	\item
	Let $M$ be a \PPT and let $x$ denote the random variable which is the output of $M$.
	Whenver the entropy of the output of $M$ is non-zero, we denote the random experiment of sampling $x$ with $x \gets M(\cdot )$. If the entropy of the output of $M$ is zero (i.e. $M$ is deterministic), we denote $x = M(\cdot)$.
	
	\item
	We sometimes think about \PPT algorithms as polynomial-size uniform families of circuits, these are equivalent models.
	A polynomial-size circuit family $\mathcal{C}$ is a sequence of circuits $\mathcal{C} = \set{C_\secp}_{\secp\in\Nat}$, such that each circuit $C_\secp$ is of polynomial size $\secp^{O(1)}$.
	We say that the family is uniform if there exists a deterministic polynomial-time algorithm $M$ that on input $1^\secp$ outputs $C_\secp$.
	
	\item
	For a \PPT algorithm $M$, we denote by $M(x;r)$ the output of $M$ on input $x$ and random coins $r$. For such an algorithm and any input $x$, we write $m\in M(x)$ to denote the fact that $m$ is in the support of $M(x;\cdot)$.
\end{itemize}

\noindent We follow standard notions from quantum computation.
\begin{itemize}
	\item
	A \QPT algorithm is a quantum polynomial-time Turing machine.
	
	\item
	We sometimes think about \QPT algorithms as polynomial-size uniform families of quantum circuits, these are equivalent models.
	A polynomial-size quantum circuit family $\mathcal{C}$ is a sequence of quantum circuits $\mathcal{C} = \set{C_\secp}_{\secp\in\Nat}$, such that each circuit $C_\secp$ is of polynomial size $\secp^{O(1)}$.
	We say that the family is uniform if there exists a deterministic polynomial-time algorithm $M$ that on input $1^\secp$ outputs $C_\secp$.
	
	\item
	An interactive algorithm $M$, in a two-party setting, has input divided into two registers and output divided into two registers.
	For the input, one register $I_m$ is for an input message from the other party, and a second register $I_a$ is an auxiliary input that acts as an inner state of the party.
	For the output, one register $O_m$ is for a message to be sent to the other party, and another register $O_a$ is again for auxiliary output that acts again as an inner state. For a quantum interactive algorithm $M$, both input and output registers are quantum.
\end{itemize}

\paragraph{The Adversarial Model.}
Throughout, efficient adversaries are modeled as quantum circuits with non-uniform quantum advice (i.e. quantum auxiliary input).
Formally, {\em a polynomial-size adversary} $\A = \set{\A_\secp, \rho_\secp}_{\secp\in\Nat}$, consists of a polynomial-size non-uniform sequence of quantum circuits $\{ \A_\secp \}_{\secp \in \Nat}$, and a sequence of polynomial-size mixed quantum states $\{ \rho_\secp \}_{\secp \in \Nat}$.

For an interactive quantum adversary in a classical protocol, it can be assumed without loss of generality that its output message register is always measured in the computational basis at the end of computation. This assumption is indeed without the loss of generality, because whenever a quantum state is sent through a classical channel then qubits decohere and are effectively measured in the computational basis.

\paragraph{Indistinguishability in the Quantum Setting.}

\begin{itemize}
	\item
	Let $f:\Nat \rightarrow [0, 1]$ be a function.
	\begin{itemize}
		\item
		$f$ is negligible if for every constant $c \in \Nat$ there exists $N \in \Nat$ such that for all $n > N$, $f(n) < n^{-c}$.
		
		\item
		$f$ is noticeable if there exists $c \in \Nat, N \in \Nat$ such that for every $n \geq N$, $f(n) \geq n^{-c}$.
		
		\item
		$f$ is overwhelming if it is of the form $1 - \mu(n)$, for a negligible function $\mu$.
	\end{itemize}
	
	\item
	We may consider random variables over bit strings or over quantum states. This will be clear from the context. 
	
	\item
	For two random variables $X$ and $Y$ supported on quantum states, quantum distinguisher circuit $\zkmD$ with, quantum auxiliary input $\rho$, and $\mu \in [0, 1]$, we write $X \approx_{\zkmD, \rho, \mu} Y$ if
	\begin{align*}
	\abs{
		\Pr[\zkmD(X; \rho)=1] - \Pr[\zkmD(Y; \rho)=1]
	} \leq \mu.
	\end{align*}
	
	\item
	Two ensembles of random variables $\mathcal{X}=\{X_{i}\}_{\secp\in \Nat, i \in I_\secp}$, $\mathcal{Y}=\{Y_{i}\}_{\secp\in \Nat, i \in I_\secp}$ over the same set of indices $I = \cupdot_{\secp \in \Nat}I_\secp$ are said to be {\em computationally indistinguishable}, denoted by $\mathcal{X}\approx_{c} \mathcal{Y}$, if for every polynomial-size quantum distinguisher $\zkmD=\set{\zkmD_\secp, \rho_\secp}_{\secp\in\Nat}$ there exists a negligible function $\mu(\cdot)$ such that for all $\secp \in \Nat, i \in I_\secp$,
	\begin{align*}
	X_i \approx_{\zkmD_{\secp}, \rho_\secp, \mu(\secp)} Y_i\enspace.
	\end{align*}
	
	\item The trace distance between two distributions $X, Y$ supported over quantum states, denoted $\TD(X, Y)$, is a generalization of statistical distance to the quantum setting and represents the maximal distinguishing advantage between two distributions supported over quantum states, by unbounded quantum algorithms.
	We thus say that ensembles $\mathcal{X}=\{X_{i}\}_{\secp\in \Nat, i \in I_\secp}$, $\mathcal{Y}=\{Y_{i}\}_{\secp\in \Nat, i \in I_\secp}$, supported over quantum states, are statistically indistinguishable (and write $\mathcal{X}\approx_{s} \mathcal{Y}$), if there exists a negligible function $\mu(\cdot)$ such that for all $\secp \in \Nat, i \in I_\secp$,
	\begin{align*}
	\TD\left( X_i, Y_i \right) \leq \mu(\secp) \enspace.
	\end{align*}
\end{itemize}

\medskip
In what follows, we introduce the cryptographic tools used in this work.
By default, all algorithms are classical and efficient, and security holds against polynomial-size non-uniform quantum adversaries with quantum advice.

\subsection{Cryptographic Tools}

\subsubsection{Interactive Proofs and Sigma Protocols}
We define interactive proof systems and then proceed to describe sigma protocols, which are a special case of interactive proof systems.
In what follows, we denote by $(\zkP, \zkV)$ a protocol between two parties $\zkP$ and $\zkV$.
For common input $\ins$, we denote by $\view_{\zkV}\prot{\zkP}{\zkV}(\ins)$ the output of $\zkV$ in the protocol. For honest verifiers, this output will be a single bit indicating acceptance or rejection of the proof.
Malicious quantum verifiers may have arbitrary quantum output.

\begin{definition}[Quantum Proof Systems for QMA] \label{def:proofs_quantum}
	Let $(\zkP, \zkV)$ be a quantum protocol with an honest \QPT prover $\zkP$ and an honest \QPT verifier $\zkV$ for a problem $\lang \in\QMA$, satisfying:
	\begin{enumerate}
		
		\item {\bf Statistical Completeness:}
		There is a polynomial $k(\cdot)$ and a negligible function $\mu(\cdot)$ s.t. for any $\secp \in\Nat$, $\ins \in \lang \cap \zo^\secp$, $\qwit \in \mathcal{R}_{\lang}(x)$\footnote{For a problem $\lang = (\langyes, \langno)$ in QMA, for an instance $x \in \langyes$, the set $\mathcal{R}_{\lang}(x)$ is the (possily infinite) set of quantum witnesses that make the BQP verification machine accept with some overwhelming probability $1 - \negl(\secp)$.},
		$$
		\Pr[ \view_{\zkV}\prot{\zkP(\qwit^{\otimes k(\secp)})}{\zkV}(\ins) = 1 ] \geq 1 - \mu(\secp) \enspace.
		$$
		
		\item {\bf Statistical Soundness:} There exists a negligible function $\mu(\cdot)$, such that for any (unbounded) prover $\zkmP$, any security parameter $\secp\in \Nat$, and any $\ins \in \zo^\secp\setminus\lang$,
		\begin{align*}
		\Pr\left[ \view_{\zkV}\prot{\zkmP}{\zkV}(\ins)=1 \right] \leq \mu(\secp)\enspace.
		\end{align*}
	\end{enumerate}
\end{definition}

We use the abstraction of {\em Sigma Protocols}, which are public-coin three-message proof systems with a weak zero-knowledge quarantee.
We define quantum Sigma Protocols for gap problems in QMA.

\begin{definition} [Quantum Sigma Protocol for QMA]
	A quantum sigma protocol for $\lang \in \QMA$ is a quantum proof system $(\qsigmaP, \qsigmaV)$ (as in Definition \ref{def:proofs_quantum}) with 3 messages and the following syntax.
	\begin{itemize}		
		\item
		$\alpha = \qsigmaP(\qwit^{\otimes k(\secp)} ; r):$ Given $k(\secp)$ copies of the quantum witness $w \in \rel_{\lang}(\ins)$ and classical randomness $r$, the first prover message consists of a quantum message $\alpha$ generated by a quantum unitary computation $\qsigmaP$.
		
		\item
		$\beta \gets \qsigmaV(x):$ The verifier simply outputs a string of $\poly(|x|)$ random bits.
		
		\item
		$\gamma = \qsigmaP_3(\beta, r):$ Given the verifier's $\beta$ and the randomness $r$, the prover outputs a response $\gamma$ by a classical computation $\qsigmaP_3$.
	\end{itemize}
	The protocol satisfies the following.
	
	\paragraph{Special Zero-Knowledge:} There exists a \QPT simulator $\qsigmaS$ such that,
	$$
	\set{(\alpha, \gamma) \; | \; r \gets U_{\ell(\secp)}, \alpha = \qsigmaP(\qwit^{\otimes k(\secp)} ; r), \gamma = \qsigmaP_3(\beta, r) }_{\secp, x, \qwit, \beta}
	$$
	$$
	\approx_{c}
	\set{ (\alpha, \gamma) \; | \; (\alpha, \gamma) \gets \qsigmaS(\ins,\beta) }_{\secp, x, \qwit, \beta} \enspace,
	$$
	where $\secp\in \Nat$, $\ins\in \lang\cap\zo^\secp$, $\qwit\in\rel_{\lang}(\ins)$, $\beta \in \{ 0, 1 \}^{\poly(\secp)}$ and $\ell(\secp)$ is the amount of randomness needed for the first prover message.
\end{definition}

\paragraph{Instantiations.}
Quantum sigma protocols follow from the parallel repetition of the 3-message quantum zero-knowledge protocols of \cite{broadbent2019zero} for QMA.

\subsubsection{Leveled Fully-Homomorphic Encryption with Circuit Privacy} \label{def:FHE}
We define a leveled fully-homomorphic encryption scheme with circuit privacy, that is, for an encryption $\ciph = \fheEnc(x)$ and a circuit $C$, a $C$-homomorphically-evaluated ciphertext $\evciph = \fheEval(C, \ciph)$ reveals nothing on $C$ but $C(x)$.

\begin{definition}[Circuit-Private Fully-Homomorphic Encryption]
	A circuit-private, leveled fully-homomoprhic encryption scheme $(\fheGen,$ $\fheEnc,$ $\fheEval,$ $\fheDec)$ has the following syntax:
	\begin{itemize}
		\item
		$\sk \gets \fheGen(1^\secp, 1^{s(\secp)}):$ a probabilistic algorithm that takes a security parameter $1^\secp$ and a circuit size bound $s(\secp)$ and outputs a secret key $\sk$.
		\item
		$\ciph \gets\fheEnc_{\sk}(x):$ a probabilistic algorithm that given the secret key, takes a string $x \in \zo^*$ and outputs a ciphertext $\ciph$.
		\item
		$\evciph \gets \fheEval(\Cir,\ciph):$ a probabilistic algorithm that takes a (classical) circuit $\Cir$ and a ciphertext $\ciph$ and outputs an evaluated ciphertext $\evciph$.
		\item
		$\hat{x} = \fheDec_{\sk}(\evciph):$ a deterministic algorithm that takes a ciphertext $\evciph$ and outputs a string $\hat{x}$.
	\end{itemize}
	
	\noindent The scheme satisfies the following.
	\begin{itemize}
		\item {\bf Perfect Correctness:} For any polynomial $s(\cdot)$, for any $\secp \in \Nat$, size-$s(\secp)$ classical circuit $C$ and input $x$ for $C$,
		$$
		\Pr\left[
		\fheDec_{\sk}(\evciph) = \Cir(x) \pST
		\begin{array}{l}
		\sk \gets \fheGen(1^\secp, 1^{s(\secp)}), \\
		\ciph \gets \fheEnc_{\sk}(x), \\
		\evciph \gets \fheEval(\Cir,\ciph)
		\end{array}
		\right] = 1\enspace.
		$$
		
		\item {\bf Input Privacy:} For every polynomial $\ell(\cdot)$ (and any polynomial $s(\secp)$),
		\begin{align*}
		\left\{
		\ciph \pST
		\begin{array}{l}
		\sk \gets \fheGen(1^\secp, 1^{s(\secp)}), \\
		\ciph \gets \fheEnc_{\sk}(x_0)
		\end{array}
		\right\}_{\secp, x_0, x_1}
		\approx_{c}
		\left\{
		\ciph \pST
		\begin{array}{l}
		\sk \gets \fheGen(1^\secp, 1^{s(\secp)}), \\
		\ciph \gets \fheEnc_{\sk}(x_1)
		\end{array}
		\right\}_{\secp, x_0, x_1} \enspace ,
		\end{align*}
		where $\secp \in \Nat$ and $x_0, x_1 \in \{ 0, 1 \}^{\ell(\secp)}$.
		
		\item {\bf Statistical Circuit Privacy:} There exist unbounded algorithms, probabilistic $\fheSim$ and deterministic $\fheExt$ such that:
		\begin{itemize}
			\item For every $x \in \{ 0, 1 \}^*$, $\ciph \in \fheEnc(x)$, the extractor outputs $\fheExt(\ciph) = x$.
			
			\item For any polynomial $s(\cdot)$,
			\begin{align*}
			\{
			\fheEval(\Cir, \ciph^*)
			\}_{\secp, \Cir, \ciph^*}
			\approx_s
			\{
			\fheSim( \; 1^\secp, \Cir(\fheExt(1^\secp, \ciph^*)) \; ) \}_{\secp, \Cir, \ciph^*}\enspace,
			\end{align*}
			where $\secp\in\Nat$, $\Cir$ is a $s(\secp)$-size circuit, and $\ciph^* \in \{ 0, 1 \}^{*}$.
		\end{itemize}
		
	\end{itemize}
\end{definition}

\noindent The next claim follows directly from the circuit privacy property, and will be used throughout the analysis.
\begin{claim}[Evaluations of Agreeing Circuits are Statistically Close] \label{claim:SFE_eval}
	For any polynomial $s(\cdot)$,
	$$
	\{ \sfeEval(C_{0}, \ciph^*) \}_{\secp, C_0, C_1, \ciph}
	\approx_s
	\{ \sfeEval(C_{1}, \ciph^*) \}_{\secp, C_0, C_1, \ciph} \enspace ,
	$$
	where $\secp \in \Nat$, $C_0$, $C_1$ are two $s(\secp)$-size functionally-equivalent circuits, and $\ciph^* \in \{ 0, 1 \}^*$.
\end{claim}

\paragraph{Instantiations.} Circuit-private leveled FHE schemes are known based on LWE \cite{ostrovsky2014maliciously, brakerski2018two}.

\subsubsection{Pseudorandom-key Public-key Encryption} \label{def:PKE}
We define a public-key encryption scheme with pseudorandom public keys.

\begin{definition}[Pseudorandom-key Public-key Encryption]
	A pseudorandom-key public-key encryption scheme $(\pkeGen,$ $\pkeEnc,$ $\pkeDec)$ has the following syntax:
	\begin{itemize}
		\item
		$(\pk, \sk) \gets \pkeGen(1^\secp):$ a probabilistic algorithm that takes a security parameter $1^\secp$ and outputs a pair of public and secret keys $(\pk, \sk)$.
		\item
		$\ciph \gets\pkeEnc_{\pk}(x):$ a probabilistic algorithm that given the public key, takes a string $x \in \zo^*$ and outputs a ciphertext $\ciph$.
		\item
		$x = \fheDec_{\sk}(\ciph):$ a deterministic algorithm that given the secret key, takes a ciphertext $\ciph$ and outputs a string $x$.
	\end{itemize}
	
	\noindent The scheme satisfies the following.
	\begin{itemize}
		\item {\bf Statistical Correctness Against Malicious Encryptors:} There is a negligible function $\negl(\cdot)$ such that for any $\secp \in \Nat$ and input $x \in \{ 0, 1 \}^*$, the following perfect correctness holds with probability at least $1 - \negl(\secp)$ over sampling $(\pk, \sk) \gets \pkeGen(1^\secp)$:
		$$
		\Pr\left[
		\pkeDec_{\sk}(\ciph) = x \pST
		\ciph \gets \pkeEnc_{\pk}(x)
		\right] = 1 \enspace.
		$$
		
		\item {\bf Public-key Pseudorandomness:} For $\secp \in \Nat$ let $\ell(\secp)$ be the length of the public key generated by $\pkeGen(1^\secp)$, then,
		\begin{align*}
		\left\{
		\pk \pST
		(\pk, \sk) \gets \pkeGen(1^\secp)
		\right\}_{\secp \in \Nat}
		\approx_{c}
		\left\{
		U_{\ell(\secp)}
		\right\}_{\secp \in \Nat}
		\enspace .
		\end{align*}
		
		\item {\bf Encryption Security:} For every polynomial $l(\cdot)$,
		\begin{align*}
		\left\{
		(\pk, \ciph) \pST
		\begin{array}{l}
		(\pk, \sk) \gets \pkeGen(1^\secp), \\
		\ciph \gets \pkeEnc_{\pk}(x_0)
		\end{array}
		\right\}_{\secp, x_0, x_1}
		\approx_{c}
		\left\{
		(\pk, \ciph) \pST
		\begin{array}{l}
		(\pk, \sk) \gets \pkeGen(1^\secp), \\
		\ciph \gets \pkeEnc_{\pk}(x_1)
		\end{array}
		\right\}_{\secp, x_0, x_1} \enspace ,
		\end{align*}
		where $\secp \in \Nat$ and $x_0, x_1 \in \{ 0, 1 \}^{l(\secp)}$.
	\end{itemize}
\end{definition}

\paragraph{Instantiations.} Pseudorandom-key public-key encryption schemes are known based on LWE \cite{Regev09}.

\subsubsection{Pseudorandom Function} \label{def:PRF}
\begin{definition} [Pseudorandom Function (PRF)] \label{def:qprf}
	A pseudorandom function scheme $(\prfGen,$ $\prfF)$ has the following syntax:
	\begin{itemize}
		\item
		$\sk \gets \prfGen(1^\secp, 1^{\ell(\secp)}):$ a probabilistic algorithm that takes a security parameter $1^\secp$ and an output size $\ell(\secp)$ and outputs a secret key $\sk$.
		\item
		$y = \prfF_{\sk}(x):$ a deterministic algorithm that given the secret key, takes a string $x \in \zo^*$ and outputs a string $y \in \{ 0, 1 \}^{\ell(\secp)}$.
	\end{itemize}
	
	\noindent The scheme satisfies the following property.
	
	\begin{itemize}
		\item {\bf Pseudorandomness:} For every quantum polynomial-size distinguisher $\zkmD = \{ \zkmD_\secp , \rho_\secp \}_{\secp \in \Nat}$ and polynomial $\ell(\cdot)$ there is a negligible function $\mu(\cdot)$ such that for all $\secp \in \Nat$,
		$$
		\abs{
			\Pr_{\sk \gets \prfGen(1^\secp, 1^{\ell(\secp)})}
			[\zkmD_\secp(\rho_\secp)^{\prfF_{\sk}(\cdot)} = 1] -
			\Pr_{f \gets (\{0, 1\}^{\ell(\secp)})^{(\{0, 1\}^*)}}
			[\zkmD_\secp(\rho_\secp)^{f(\cdot)} = 1]
		} \leq \mu(\secp) \enspace .
		$$
	\end{itemize}
\end{definition}

\subsubsection{NIZK Argument for NP in the Common Random String Model} \label{def:NIZK}
We define non-interactive computational zero-knowledge arguments for NP in the common random string model, with adaptive multi-theorem security.

\begin{definition}[NICZK Argument for NP] \label{def:nizk_np}
	A non-interactive computational zero-knowledge argument system in the common random string model for a language $\lang \in \NP$ consists of 3 algorithms $(\nizkSet$ $, \nizkP$ $, \nizkV)$ with the following syntax:
	\begin{itemize}
		\item
		$\crs \gets \nizkSet(1^\secp) : $
		A classical algorithm that on input security parameter $\secp$ simply samples a common uniformly random string $\crs$.
		
		\item
		$\pi \gets \nizkP(\crs , x , \wit) : $
		A probabilistic algorithm that on input $\crs$, an instance $x \in \lang$ and a witness $\wit \in \mathcal{R}_{\lang}(\ins)$, outputs a proof $\pi$.
		
		\item
		$\nizkV(\crs, x, \pi) \in \{ 0 , 1 \} : $
		A deterministic algorithm that on input $\crs$, an instance $x \in \lang$ and a proof $\pi$, outputs a bit.
	\end{itemize}
	The protocol satisfies the following properties.
	\begin{itemize}
		\item {\bf Perfect Completeness:}
		For any $\secp \in\Nat$, $\ins \in \lang \cap \zo^\secp$, $\wit \in \mathcal{R}_{\lang}(x)$,
		$$
		\Pr_{\substack{
				\crs \gets \nizkSet(1^\secp),\\
				\pi \gets \nizkP(\crs , x , \wit)
		}}
		\Big[ \nizkV(\crs, x, \pi) = 1 \Big] = 1 \enspace.
		$$
		
		\item {\bf Adaptive Computational Soundness:}
		For every quantum polynomial-size prover $\nizkP^* = \{ \nizkP^*_\secp , \rho_\secp \}_{\secp \in \Nat}$ there is a negligible function $\mu(\cdot)$ such that for every security parameter $\secp \in \Nat$,
		$$
		\Pr_{\substack{
				\crs \gets \nizkSet(1^\secp),\\
				(x, \pi^*) \gets \nizkP^*_\secp(\rho_\secp, \crs)
		}}
		\Big[
		(x \notin \lang) \land \big( 1 = \nizkV(\crs, x, \pi^*) \big)
		\Big] \leq \mu(\secp) \enspace .
		$$
		
		\item {\bf Multi-Theorem Adaptive Computational Zero Knowledge:}
		There exists a polynomial-time simulator $\nizkSim$ such that for every quantum polynomial-size distinguisher $\Disting = \{ \Disting_\secp , \rho_\secp \}_{\secp \in \Nat}$ there is a negligible function $\mu(\cdot)$ such that for every security parameter $\secp \in \Nat$,
		$$
		\abs{P_{\secp, \real} - P_{\secp, \simulated}} \leq \mu(\secp) \enspace ,
		$$
		where,
		$$
		P_{\secp, \real} :=
		\Pr_{\crs \gets \nizkSet(1^\secp)}
		\Big[
		\Disting_{\secp}(\rho_{\secp}, \crs)^{\nizkP(\crs, \cdot, \cdot)} = 1
		\Big] \enspace ,
		$$
		$$
		P_{\secp, \simulated} :=
		\Pr_{(\zkcrs, \td) \gets \nizkSim(1^\secp)}
		\Big[
		\Disting_{\secp}(\rho_{\secp}, \zkcrs)^{\nizkSim(\td, \cdot)} = 1
		\Big] \enspace ,
		$$
		where,
		\begin{itemize}
			\item
			In every query that $\Disting$ makes to the oracle, it sends a pair $(x, \wit)$ where $x \in \lang\cap \{0, 1\}^\secp$ and $\wit \in \rel_{\lang}(x)$.
			
			\item
			$\nizkP(\crs, \cdot, \cdot)$ is the prover algorithm and $\nizkSim(\cdot, \cdot)$ acts only on its sampled trapdoor $\td$ and on $x$.
		\end{itemize}
		
	\end{itemize}
	
\end{definition}

\paragraph{Instantiations.} Non-interactive computational zero-knowledge arguments for NP in the common random string model with both adaptive soundness and zero knowledge are known based on LWE \cite{canetti2019fiat, peikert2019noninteractive}.

\subsubsection{Malicious Designated-Verifier Non-interactive Zero-knowledge for QMA}
We define non-interactive zero-knowledge protocols in the malicious designated-verifier model (MDV-NIZK) for QMA, with adaptive (and non-adaptive) multi-theorem security.

\begin{definition}[MDV-NICZK Argument for QMA] \label{def:mdv_nizk_qma}
	A non-interactive computational zero-knowledge argument system for in the malicious designated-verifier model for a gap problem $(\langyes, \langno) = \lang \in \QMA$ consists of 4 algorithms $(\zkSet$ $, \zkVSet$ $, \zkP$ $, \zkV)$ with the following syntax:
	\begin{itemize}
		\item
		$\crs \gets \zkSet(1^\secp) : $
		A classical algorithm that on input security parameter $\secp$ simply samples a common uniformly random string $\crs$.
		
		\item
		$(\pvk , \svk) \gets \zkVSet(\crs) : $
		A classical algorithm that on input $\crs$ samples a pair of public and secret verification keys.
		
		\item
		$\ket{\pi} \gets \zkP(\crs , \pvk , x , \qwit^{\otimes k(\secp)}) : $
		A quantum algorithm that on input $\crs$, the public verification key $\pvk$, an instance $x \in \langyes$ and polynomially-many identical copies of a witness $\qwit \in \mathcal{R}_{\lang}(\ins)$ ($k(\cdot)$ is some polynomial), outputs a quantum state $\ket{\pi}$.
		
		\item
		$\zkV(\crs, \svk, x, \ket{\pi}) \in \{ 0 , 1 \} : $
		A quantum algorithm that on input $\crs$, secret verification key $\svk$, an instance $x \in \lang$ and a quantum proof $\ket{\pi}$, outputs a bit.
	\end{itemize}
	The protocol satisfies the following properties.
	\begin{itemize}
		\item {\bf Statistical Completeness:}
		There is a polynomial $k(\cdot)$ and a negligible function $\mu(\cdot)$ s.t. for any $\secp \in\Nat$, $\ins \in \langyes \cap \zo^\secp$, $\qwit \in \mathcal{R}_{\lang}(x)$, $\crs \in \zkSet(1^\secp)$, $(\pvk, \svk) \in \zkVSet(\crs)$,
		$$
		\Pr_{\ket{\pi} \gets \zkP(\crs , \pvk , x , \qwit^{\otimes k(\secp)})}
		\Big[ \zkV(\crs, \svk, x, \ket{\pi}) = 1 \Big] \geq 1 - \mu(\secp) \enspace.
		$$
		
		\item {\bf Multi-Theorem Adaptive Computational Soundness:}
		For every quantum polynomial-size prover $\zkmP = \{ \zkmP_\secp , \rho_\secp \}_{\secp \in \Nat}$ there is a negligible function $\mu(\cdot)$ such that for every security parameter $\secp \in \Nat$,
		$$
		\Pr_{\substack{
				\crs \gets \zkSet(1^\secp),\\
				(\pvk, \svk) \gets \zkVSet(\crs),\\
				(x, \ket{\pi^*}) \gets \zkmP_\secp(\rho_\secp, \crs, \pvk)^{\zkV(\crs, \svk, \cdot, \cdot)}
			}}
		\Big[
		(x \in \lang_{no}) \land \big( 1 = \zkV(\crs, \svk, x, \ket{\pi^*}) \big)
		\Big] \leq \mu(\secp) \enspace .
		$$
		
		\item {\bf Multi-Theorem Adaptive Computational Zero Knowledge:}
		There exists a quantum polynomial-time simulator $\zkSim$ such that for every quantum polynomial-size distinguisher $\Disting = \{ \Disting_\secp , \rho_\secp \}_{\secp \in \Nat}$ there is a negligible function $\mu(\cdot)$ such that for every security parameter $\secp \in \Nat$,
		$$
		\abs{
			\Pr_{\crs \gets \zkSet(1^\secp)}
			\Big[
			\Disting_{\secp}(\rho_{\secp}, \crs)^{\zkP(\crs, \cdot, \cdot, \cdot)} = 1
			\Big]
			-  
			\Pr_{(\zkcrs, \td) \gets \zkSim(1^\secp)}
			\Big[
			\Disting_{\secp}(\rho_{\secp}, \zkcrs)^{\zkSim(\td, \cdot, \cdot)} = 1
			\Big]
		}
		\leq \mu(\secp) \enspace ,
		$$
		where,
		\begin{itemize}
			\item
			In every query that $\Disting$ makes to the oracle, it sends a triplet $(\mpvk, x, \qwit^{\otimes k(\secp)})$ where $\mpvk$ can be arbitrary, $x \in \langyes\cap \{0, 1\}^\secp$ and $\qwit \in \rel_{\lang}(x)$.
			
			\item
			$\zkP(\crs, \cdot, \cdot, \cdot)$ is the prover algorithm and $\zkSim(\cdot, \cdot)$ acts only on its sampled trapdoor $\td$ and on $\mpvk, x$.
		\end{itemize}
		
	\end{itemize}

\end{definition}

	We note that the standard (non-adaptive) soundness guarantees the following: 
	\begin{definition}[MDV-NICZK Argument for QMA with Standard Soundness] \label{def:mdv_nizk_qma_sound}
		A non-interactive computational zero-knowledge argument system in the malicious designated-verifier model for a gap problem $(\langyes, \langno) = \lang \in \QMA$ has standard non-adaptive soundness if it satisfies the same properties described in definition \ref{def:mdv_nizk_qma}, with the only change that instead of satisfying multi-theorem adaptive soundness, it satisfies the following guarantee:
		\begin{itemize}
			\item {\bf Multi-Theorem Computational Soundness:}
			For every quantum polynomial-size prover $\zkmP = \{ \zkmP_\secp , \rho_\secp \}_{\secp \in \Nat}$ and $\{ x_\secp \}_{\secp \in \Nat}$ where $\forall \secp \in \Nat: x_\secp \in \langno$, there is a negligible function $\mu(\cdot)$ such that for every security parameter $\secp \in \Nat$,
			$$
			\Pr_{\substack{
					\crs \gets \zkSet(1^\secp),\\
					(\pvk, \svk) \gets \zkVSet(\crs),\\
					\ket{\pi^*} \gets \zkmP_\secp(\rho_\secp, \crs, \pvk)^{\zkV(\crs, \svk, \cdot, \cdot)}
			}}
			\Big[
			1 = \zkV(\crs, \svk, x, \ket{\pi^*})
			\Big] \leq \mu(\secp) \enspace .
			$$
		\end{itemize}
		
	\end{definition}

\section{Non-interactive Zero-knowledge Protocol}
In this section we describe a non-interactive computational zero-knowledge argument system in the malicious designated-verifier model for an arbitrary $\lang \in \QMA$, according to Definition \ref{def:mdv_nizk_qma}.

	\paragraph{Ingredients and notation:}
	\begin{itemize}
		\item
		A non-interactive zero-knowledge argument for NP $(\nizkSet,$ $\nizkP,$ $\nizkV)$ in the common random string model.
		\item
		A pseudorandom function $(\prfGen, \prfF)$.
		\item
		A leveled fully-homomorphic encryption scheme $(\fheGen,$ $\fheEnc,$ $\fheEval,$ $\fheDec)$ with circuit privacy.
		\item
		A public-key encryption scheme $(\pkeGen, \pkeEnc, \pkeDec)$ with pseudorandom public keys.
		\item
		A 3-message quantum sigma protocol $(\qsigmaP, \qsigmaV)$ for QMA.
	\end{itemize}
	
	\noindent We describe the protocol in \figref{fig:nizk_qma}.
	
	\protocol
	{\proref{fig:nizk_qma}}
	{A non-interactive computational zero-knowledge argument system for $\lang \in \QMA$ in the malicious designated-verifier model.}
	{fig:nizk_qma}
	{
		\begin{description}
			\item[Common Input:] An instance $\ins \in \langyes\cap \zo^\secp$, for security parameter $\secp \in \Nat$.
			
			\item[$\zkP$'s private input:] Polynomially many identical copies of a witness for $\ins$: $\ket{w}^{\otimes k(\secp)}$ s.t. $\ket{w} \in \mathcal{R}_{\lang}(\ins)$.
		\end{description}
		
		\begin{enumerate}
			\item {\bf Common Random String:}
			$\zkSet$ samples the common random string of the NP NIZK argument, $\crs \gets \nizkSet(1^\secp)$ and an additional random string $\ek \gets U_{\ell(\secp)}$ where $\ell(\secp)$ is the size of a public key generated by $\pkeGen(1^{\secp})$.
			$\zkSet$ publishes $(\crs , \ek)$ as the common random string.
			
			\item {\bf Public and Secret Verification Keys:}
			$\zkVSet$ samples public and secret verification keys:
			\begin{itemize}
				\item
				Samples $\prfk \gets \prfGen(1^\secp)$, $\fhek \gets \fheGen(1^\secp)$ and encrypts the PRF key using the FHE encryption, $\ciph_{\zkV} \gets \fheEnc_{\fhek}(\prfk)$.
				
				\item
				Let $r_\zkV$ be the randomness used for $\prfGen, \fheGen, \fheEnc$. $\zkVSet$ encrypts $\ciph_{r_\zkV} \gets \pkeEnc_{\ek}(r_\zkV)$ and computes a NIZK proof $\pi_{\zkV} \gets \nizkP(\crs, (\ciph_{\zkV}, \ciph_{r_\zkV}, \ek))$, for the NP statement declaring that the tuple $(\ciph_{\zkV}, \ciph_{r_\zkV}, \ek)$ is consistent.\footnote{Formally, there exist $r_1, r_2$ s.t. $\ciph_{\zkV}$ is generated by using $\prfGen, \fheGen, \fheEnc$ with randomness $r_1$, and $\ciph_{r_\zkV} = \pkeEnc_{\ek}(r_1 ; r_2)$.}		
				
			\end{itemize}		
			The key values are: $\pvk = (\ciph_{\zkV}, \ciph_{r_\zkV}, \pi_\zkV)$, $\svk = (\prfk, \fhek)$.
			
			\item {\bf Non-interactive Zero-knowledge Proof:}
				Given $(\crs, \ek)$ and $\pvk$, $\zkP$ first checks that $1 = \nizkV(\crs, (\ciph_{\zkV}, \ciph_{r_\zkV}, \ek), \pi_\zkV)$ and aborts otherwise.
			\begin{itemize}
				\item
				$\zkP$ computes the sigma protocol message $\alpha = \qsigmaP(\ket{w}^{\otimes k(\secp)} ; r_\Xi)$, for randomness $r_\Xi$.
				
				\item
				$\zkP$ computes $\evciph_\zkP \gets \fheEval(C_{x, r_\Xi} , \ciph_{\zkV})$, where $C_{x, r_\Xi}$ is the following circuit:
				Given input $\prfk$ a PRF secret key, $C_{x, r_\Xi}$ computes $\beta_x = \prfF_{\prfk}(x)$, and then outputs $\gamma = \qsigmaP_3(\beta_x, r_\Xi)$.
				
				\item
				$\zkP$ encrypts $\ciph_{r_\Xi} \gets \pkeEnc_{\ek}(r_\Xi)$ and computes a NIZK proof $\pi_{\zkP} \gets \nizkP(\crs, (\evciph_{\zkP}, \ciph_{r_\Xi}, \ek))$, for the NP statement declaring that the tuple $(\evciph_{\zkP}, \ciph_{r_\Xi}, \ek)$ is consistent.\footnote{Formally, there exist $r_\Xi, r_1, r_2$ s.t. $\evciph_{\zkP} = \fheEval(C_{x, r_\Xi} , \ciph_{\zkV} ; r_1)$, $\ciph_{r_\Xi} = \pkeEnc_{\ek}(r_\Xi ; r_2)$.}
			\end{itemize}
			$\zkP$ sends $\ket{\pi} = (\alpha, \evciph_{\zkP}, \ciph_{r_\Xi}, \pi_\zkP)$ to $\zkV$.
			
			\item {\bf Verification:} \label{nizk:ver}
			Given $(\crs, \ek)$, $\svk$ and $\ket{\pi}$, $\zkV$ accepts iff all of the following holds:
			\begin{itemize}
				\item
				$1 = \nizkV(\crs, (\evciph_\zkP, \ciph_{r_\Xi}, \ek), \pi_\zkP)$.
				
				\item
				Let $\beta_x = \prfF_{\prfk}(x)$, $\gamma = \fheDec_{\fhek}(\evciph_{\zkP})$, then $1 = \qsigmaV(x, \alpha, \beta_x, \gamma)$.
			\end{itemize}
			
		\end{enumerate}
	}

	The (statistical) completeness of the protocol follows readily from the perfect completeness of the $\nizk$ scheme, the perfect correctness of $\fhe$ and the statistical completeness of the quantum sigma protocol $(\qsigmaP, \qsigmaV)$.
	We next prove the soundness and zero knowledge of the protocol.

	\subsection{Soundness}
	We prove that the protocol has multi-theorem computational soundness (as in Definition \ref{def:mdv_nizk_qma_sound}).
	By standard generic compilation and sub-exponential hardness of LWE we extend our soundness to be adaptive (as in Definition \ref{def:mdv_nizk_qma}).
%	 i.e. the prover is unable to fake a proof even if it can choose $x \in \langno$ \emph{after} seeing the CRS, the public verification key and getting oracle access to the verifier's responses for proofs
	
	\begin{proposition} [The Protocol has Multi-theorem Computational Soundness] \label{prop:soundness}
		For every quantum polynomial-size prover $\zkmP = \{ \zkmP_\secp , \rho_\secp \}_{\secp \in \Nat}$ there is a negligible function $\mu(\cdot)$ such that for every security parameter $\secp \in \Nat$ and $x \in \langno \cap \{ 0, 1 \}^\secp$,
		$$
		\Pr_{\substack{
				(\crs, \ek) \gets \zkSet(1^\secp),\\
				\big( (\ciph_{\zkV}, \ciph_{r_\zkV}, \pi_\zkV), (\prfk, \fhek) \big) \gets \zkVSet(\crs, \ek),\\
				\ket{\pi^*} \gets \zkmP_\secp\big( \rho_\secp, (\crs, \ek), (\ciph_{\zkV}, \ciph_{r_\zkV}, \pi_\zkV) \big)^{\zkV((\crs, \ek), (\prfk, \fhek), \cdot, \cdot)}
		}}
		\Big[
		1 = \zkV((\crs, \ek), (\prfk, \fhek), x, \ket{\pi^*})
		\Big] \leq \mu(\secp) \enspace .
		$$
	\end{proposition}
	
	\begin{proof}
		Let $\zkmP = \{ \zkmP_\secp, \rho_\secp \}_{\secp \in \Nat}$ a polynomial-size quantum prover and let $\{ x_\secp \}_{\secp \in \Nat}$ s.t. $\forall \secp \in \Nat : x_\secp \in \langno \cap \{ 0, 1 \}^{\secp}$.
		We prove soundness by a hybrid argument, that is, we consider a series of computationally-indistinguishable hybrid processes with output over $\{ 0, 1 \}$, starting from the output of the verifier (for the prover's false proof) in the real inteaction, until we get to a distribution where the output of the verifier can be $1$ with at most negligible probability.
		We define the following processes.
		\begin{itemize}
			\item $\Hyb_0:$
			The output distribution of the verifier in the real interaction, that is, for
			$$
			(\crs, \ek) \gets \zkSet(1^\secp) \enspace ,
			\big( (\ciph_{\zkV}, \ciph_{r_\zkV}, \pi_\zkV), (\prfk, \fhek) \big) \gets \zkVSet(\crs, \ek) \enspace , 
			$$
			$$
			\ket{\pi^*} \gets \zkmP_\secp\big( \rho_\secp, (\crs, \ek), (\ciph_{\zkV}, \ciph_{r_\zkV}, \pi_\zkV) \big)^{\zkV((\crs, \ek), (\prfk, \fhek), \cdot, \cdot)} \enspace ,
			$$
			the output bit $\zkV((\crs, \ek), (\prfk, \fhek), x, \ket{\pi^*})$.
			
			\item $\Hyb_1:$
			This hybrid process is identical to $\Hyb_0$, with the exception that $\ek$ is sampled as a public key for the PKE scheme $(\ek, \sk) \gets \pkeGen(1^\secp)$, rather than as a random string of the same length.
			To move to this hybrid we will use the fact that the public keys of the PKE scheme are pseudorandom.
			
			\item $\Hyb_2:$
			This hybrid process is identical to $\Hyb_1$, with the exception that the verification algorithm (described in step \ref{nizk:ver} of the protocol) changes.
			The new verifier $\tilde{\zkV}$ still makes sure that $\pi_\zkP$ is a valid proof for $(\evciph_{\zkP}, \ciph_{r_\Xi}, \ek)$, but the second check changes to the following:
			Let $r_{\Xi} = \pkeDec_{\sk}(\ciph_{r_\Xi})$, and let $\gamma = \qsigmaP_3(\beta_x, r_{\Xi})$. Then $\tilde{\zkV}$ accepts if $1 = \qsigmaV(x, \alpha, \beta_x, \gamma)$.
			To move to this hybrid we will use the (adaptive) soundness property of the NP NIZK proof that $\zkmP$ provides.
			
			\item $\Hyb_3:$
			This hybrid process is identical to $\Hyb_2$, with the exception that when generating the CRS $(\crs, \ek)$ and the public verification key $\pvk = (\ciph_{\zkV}, \ciph_{r_\zkV}, \pi_\zkV)$, (1) the CRS for the NP NIZK is simulated $(\crs, \td) \gets \nizkSim(1^\secp)$, (2) the proof $\pi_\zkV$ is simulated $\pi_\zkV \gets \nizkSim( \td, (\ciph_{\zkV}, \ciph_{r_\zkV}, \ek) )$ rather than generated by the NP NIZK prover.
			To move to this hybrid we use the zero-knowledge property of the NP NIZK proof that $\zkV$ provides.
			
			\item $\Hyb_4:$
			This hybrid process is identical to $\Hyb_3$, with the exception that when generating $\pvk = (\ciph_{\zkV}, \ciph_{r_\zkV}, \pi_\zkV)$, $\ciph_{r_\zkV}$ is just an encryption of a string of zeros (of the same length) rather than the randomness $r_\zkV$.
			To move to this hybrid we use the security of the PKE scheme.
			
			\item $\Hyb_5:$
			This hybrid process is identical to $\Hyb_4$, with the exception that when generating $\pvk = (\ciph_{\zkV}, \ciph_{r_\zkV}, \pi_\zkV)$, $\ciph_{\zkV}$ is just an encryption of a string of zeros (of the same length) rather than the FHE encryption of the secret PRF key $\prfk$.
			To move to this hybrid we use the security of the FHE scheme.
			
			\item $\Hyb_6:$
			This hybrid process is identical to $\Hyb_5$, with the exception that the modified verification algorithm $\tilde{\zkV}$ from $\Hyb_2$ is now going to be a new \emph{stateful} algorithm $\tilde{\zkV}_s$.
			The new verifier $\tilde{\zkV}_s$ still makes sure that $\pi_\zkP$ is a valid proof for $(\evciph_{\zkP}, \ciph_{r_\Xi}, \ek)$, but the second check changes to the following:
			It is identical to that of $\tilde{\zkV}$, except that $\beta_x$ is now lazily sampled as a truly random string, that is, every time $\zkmP$ sends a query for some $x'$, instead of computing $\beta_{x'} = \prfF_{\prfk}(x')$, $\tilde{\zkV}_s$ samples $\beta_{x'}$ a truly random string of the same length and remembers it for future queries by the prover (for the same $x'$).	
			To move to this hybrid we use the pseudorandomnes guarantee of the PRF.
			
			\item $\Hyb_7:$
			This hybrid process is identical to $\Hyb_6$, with the exception that the behaviour of the verification algorithm $\tilde{\zkV}_s$ changes in the following way: Consider $t$ the \emph{first} time step in the execution of $\zkmP$ (in $\Hyb_6$) such that with a noticeable probability, $\zkmP$ sends a pair $(x', \ket{\pi^*})$ such that (1) $x' \in \langno$ and (2) the modified verification algorithm $\tilde{\zkV}_s$ accepts - this proof can be sent either as a query to the verification oracle, or as the final output of $\zkmP$ (in that case, $t$ is the last time step of $\zkmP$ and $x' = x$).
			
			Now we define $\Hyb_7$: the verification algorithm works as in $\Hyb_6$ with the one change that if $\zkmP$ sends a query to the verification oracle before its time step $t$ and this query is for a no-instance $x' \in \langno$, then we simply return $0$ to $\zkmP$ as the verifier's answer, without computing anything.
			Note that checking whether $x' \in \langno$ takes $2^{O(|x'|)}$ time\footnote{We assume that our gap problem $\lang \in \QMA$ has exponential-time algorithms that solve it, that is, for $x \in \lang$ we can decide whether $x \in \langyes$ or $x \in \langno$ in $2^{O(|x|)}$ time. It is also enough for our proof to assume that $\lang$ is solvable in general exponential time i.e. $O(2^{|x|^c})$ time for some constant $c \in \Nat$.}, and thus the execution of this hybrid is inefficient.
			If such time step $t$ does not exist (i.e. in each of the prover's time steps, the probability for it to generate a false proof is only negligible), this process is identical to $\Hyb_6$.
		\end{itemize}
		
		We now explain why the outputs of each two consecutive hybrids are computationally indistinguishable\footnote{the output bits of the hybrids are in fact \emph{statistically} indistinguishable, because any two distributions over a bit are statistically indistinguishable if they are computationally indistinguishable, but we won't care about this in our analysis.}.
		We will then use the last hybrid process to show that soundness of the protocol follows from the soundness of the quantum sigma protocol $(\qsigmaP, \qsigmaV)$.
		\begin{itemize}
			\item $\Hyb_0 \approx_{c} \Hyb_1:$
			Follows readily from the pseudorandomness property of the public keys generated by $\pkeGen(1^\secp)$.

			\item $\Hyb_1 \approx_{c} \Hyb_2:$
			Follows from the adaptive soundness of the NIZK protocol for NP, the statistical correctness of the PKE scheme and the perfect correctness of the FHE scheme.
			We explain in more detail: Assume the output bits of $\Hyb_1$ and $\Hyb_2$ are distinguishable with some noticeable advantage, then by the perfect correctness of the FHE evaluation, it follows that with a noticeable probability, either (1) there was an error in the decryption process of the PKE scheme at least once, or (2) $\zkmP$ generated a false proof for the NP NIZK scheme at least once.
			We prove that both happen with at most negligible probability, and thus the statistical distance between the output bits of $\Hyb_1$ and $\Hyb_2$ is at most negligible.
			
			The correctness guarantee of the PKE scheme is that when the public key is sampled honestly, which is true in our case, then with overwhelming probability over the randomness of $\pkeGen(1^\secp)$, the decryption is perfectly correct, regardless of the randomness used for the encryption (which in our case is possibly malicious, as it is chosen by $\zkmP$).
			This implies that with at most negligible probability there is an error in the decryption process $\pkeDec_{\sk}(\cdot)$.
			
			If $\zkmP$ manages to give a false proof $\pi_{\zkP}^*$ for some tuple $(\evciph_{\zkP}, \ciph_{r_\Xi}, \ek)$ with a noticeable probability $\varepsilon$ then we can use it to break the adaptive soundness of the NP NIZK scheme: We guess the index of the query (to the verification oracle $\tilde{\zkV}((\crs, \ek), (\prfk, \fhek), \cdot, \cdot)$) where $\zkmP$ gives such false proof, and with probability at least $\varepsilon \cdot \frac{1}{t}$, where $t$ is the (polynomial) running time of $\zkmP$, we find such false proof.
			This implies that $\varepsilon$ has to be at most negligible i.e. $\zkmP$ cannot produce a false proof for the NP NIZK with a noticeable probability.
			
			\item $\Hyb_2 \approx_{c} \Hyb_3:$
			Assume toward contradiction that the output bits of $\Hyb_2$ and $\Hyb_3$ are distinguishable with some noticeable advantage, we use the prover $\zkmP$ in order to construct a distinguisher $\zkmD$ that breaks the zero-knowledge property of the NP NIZK scheme (it seems that we don't have to use the fact that the zero knowledge property of the NP NIZK is adaptive, but we will use it for the convenience of the proof and because it does not cause an extra cost in computational assumptions).
			
			$\zkmD$ will sample $(\ek, \sk) \gets \pkeGen(1^\secp)$, honestly sample $(\ciph_{\zkV}, \ciph_{r_\zkV})$ with randomness $r$, and then get a common random string $\crs$ from the NIZK zero knowledge challenger.
			$\zkmD$ then hands $(\ciph_{\zkV}, \ciph_{r_\zkV}, \ek)$ along with the NP witness $r$ and gets back either a real proof or a simulated proof. it then proceeds to run the malicious prover $\zkmP$ and at the end, by the verdict of the (modified) verification algorithm $\tilde{\zkV}$ for the prover's proof and instance, distinguishes between whether it got a simulated proof or a real proof. This follows from the fact that whenever $\zkmD$ gets a real proof (and CRS) then the view of $\zkmP$ is exactly its view in $\Hyb_2$ and whenever $\zkmD$ gets a simulated proof (and CRS) then the view of $\zkmP$ is exactly its view in $\Hyb_3$.

			\item $\Hyb_3 \approx_{c} \Hyb_4:$
			Follows readily from the security of the PKE scheme.
			
			\item $\Hyb_4 \approx_{c} \Hyb_5:$
			Follows readily from the security of the FHE scheme.
			
			\item $\Hyb_5 \approx_{c} \Hyb_6:$
			Follows readily from the security of the PRF scheme.
			
			\item $\Hyb_6 \approx_{c} \Hyb_7:$
			Note that by how we defined the time step $t$ it follows that the change of returning $0$ on queries for no-instances before time step $t$ (rather than actually evaluating the verification algorithm $\tilde{\zkV}_s$) is unnoticeable to the prover $\zkmP$.
		\end{itemize}
	
		Now, assume toward contradiction that $\zkmP$ succeeds in breaking the soundness with a noticeable probability in the original execution of the protocol (i.e. in the process $\Hyb_0$), and by the fact $\Hyb_0 \approx_{c} \Hyb_7$ it follows that the verifier accepts the prover's false proof with some noticeable probability in the hybrid experiment $\Hyb_7$.
		By the fact that with some noticeable probability $\zkmP$ succeeds in cheating in $\Hyb_7$, it follows that a time step $t$ exists where $\zkmP$ sends a pair $(x', \ket{\pi^*})$ such that $x' \in \langno$ and $\tilde{\zkV}_s$ accepts the proof (this follows because in the last step of $\zkmP$'s execution it sends noticeably often a successful false proof for $x \in \langno$).
		
		Now we consider the execution process of $\Hyb_7$ and fix by an averaging argument the snapshot $\ket{\psi}$ of the execution in the exact moment where $\zkmP$ sends a pair $(x', \ket{\pi^*})$ in its time step $t$, such that the snapshot maximizes the probability that $x' \in \langno$ and $\tilde{\zkV}_s$ accepts the proof $\ket{\pi^*}$ (as a side note, this snapshot includes (1) all of the randomness (including setup information) in the process $\Hyb_7$ until $\zkmP$'s step $t$, (2) the inner quantum state of $\zkmP$ in step $t$, and of course a pair $(x', \ket{\pi^*})$ such that $x' \in \langno$.).
		It follows that the part $\alpha$ and the extracted $\gamma$ (both obtained from $\ket{\pi^*}$, recall $\gamma$ is obtained by the extracted randomness $r_{\Xi}$ and the random string $\beta_{x'}$) make a quantum sigma protocol verifier $\qsigmaV$ accept the proof for a random challenge $\beta$ with a noticeable probability.
		
		We now describe a malicious prover $\qsigmaP^*$ that breaks the soundness of the quantum sigma protocol $(\qsigmaP, \qsigmaV)$, by using $\zkmP$ and the quantum advice $\ket{\psi}$ in order to convince $\qsigmaV$ to accept the no-instance $x' \in \langno$.
		$\qsigmaP^*$ uses the snapshot $\ket{\psi}$ and takes $\alpha$ from $\ket{\pi^*}$ and sends it as the first sigma protocol message to $\qsigmaV$. $\qsigmaV$ returns a random challenge $\beta$, and $\qsigmaP^*$ treats this random challenge as the random $\beta_{x'}$ for the verification procedure $\tilde{\zkV}_{s}$.
		$\qsigmaP^*$ then derives $\gamma$ from $\ket{\pi^*}$ (as usual in $\tilde{\zkV}_s$) and sends it to $\qsigmaV$.
		Recall that we know $\qsigmaV$ accepts the proof with a noticeable probability, and thus $\qsigmaP^*$ breaks the soundness of the quantum sigma protocol with noticeable probability, in contradiction.
	\end{proof}

	We next use standard complexity leveraging to make the soundness adaptive, that is, by assuming that the security of our cryptographic primitives is sub-exponential we prove that the prover cannot choose the no-instance $x \in \langno$ adaptively. As mentioned in the preliminaries, the security of all of our primitives can be based on the hardness of LWE, and thus based on the sub-exponential hardness of LWE we can get adaptive soundness.
	
	\begin{proposition} [The Protocol has Multi-theorem Adaptive Computational Soundness] \label{prop:adaptive_soundness}
		Assume there is a constant $\varepsilon \in (0, 1)$ such that the cryptographic ingridients we use are secure against $O(2^{\secp^{\varepsilon}})$-time quantum algorithms for security paramter $\secp$.
		Then, by executing the protocol with security parameter $\secp := |x|^{\frac{2}{\varepsilon}}$ rather than $\secp = |\ins|$, for every quantum polynomial-size prover $\zkmP = \{ \zkmP_{\secp} , \rho_{\secp} \}_{\secp \in \Nat}$ there is a negligible function $\mu(\cdot)$ such that for every security parameter $\secp \in \Nat$,
				$$
				\Pr
				\Big[
				(x \in \lang_{no}) \land \big( 1 = \zkV((\crs, \ek), (\prfk, \fhek), x, \ket{\pi^*}) \big)
				\Big] \leq \mu(\secp) \enspace ,
				$$
				where the probability is above the following experiment:
				$$
				(\crs, \ek) \gets \zkSet(1^{\secp}), \enspace 
				\big( (\ciph_{\zkV}, \ciph_{r_\zkV}, \pi_\zkV), (\prfk, \fhek) \big) \gets \zkVSet(\crs, \ek),
				$$
				$$
				(x, \ket{\pi^*}) \gets \zkmP_\secp\big( \rho_\secp, (\crs, \ek), (\ciph_{\zkV}, \ciph_{r_\zkV}, \pi_\zkV) \big)^{\zkV((\crs, \ek), (\prfk, \fhek), \cdot, \cdot)} \enspace .
				$$
	\end{proposition}
	
	\begin{proof}
		The proof is almost identical to the proof of Proposition \ref{prop:soundness}, with minor technical changes.
		Let $\zkmP = \{ \zkmP_\secp, \rho_\secp \}_{\secp \in \Nat}$ a polynomial-size quantum prover in \proref{fig:nizk_qma} and as before, we prove soundness by a hybrid argument by considering almost the same series of hybrids processes, and the reductions that show the outputs of each consecutive pair of hybrids are indistinguishable, are also going to be slightly different.
		
		More precisely, consider the exact same hybrids $\Hyb_0, \cdots, \Hyb_7$ from the proof of Proposition \ref{prop:soundness}, with only the following differences:
		\begin{itemize}
			\item
			With accordance to the fact that we consider adaptive provers, in each hybrid process, the output of the malicious prover at the end of the execution is a pair $(x, \ket{\pi^*})$ rather than only a proof $\ket{\pi^*}$.
			
			\item
			The output of each hybrid process is still a bit, but going to be the logical AND of (1) the verifier accepting the prover's proof and instance $x$, and (2) the instance $x$ is indeed a no-instance $x \in \langno$ (note that in the proof for Proposition \ref{prop:soundness} the output bit of the hybrids only considers the verdict of the verifier, as the no-instance $x \in \langno$ is already fixed).
		\end{itemize}
		
		We will next claim that the outputs of each pair of consecutive hybrids are computationally indistinguishable. For this, we will use the fact that given $x \in \lang = \langyes \cup \langno$, we can decide whether $x \in \langno$ or not in $2^{O(|x|)}$ time.\footnote{As noted before, the proof is not sensitive to the fact that the time complexity is $2^{O(|x|)}$ and not $O(2^{|x|^c})$ time for some constant $c \in \Nat$.}
		We also use the fact that our primitives are assumed to be secure against sub-exponential time algorithms and we run the protocol with increased security parameter, more specifically, we assume that our primitives are secure against $O(2^{\secp^{\varepsilon}})$-time algorithms and we use security paramter $\secp = |x|^{\frac{2}{\varepsilon}}$, thus it follows that no $O(2^{\secp^{\varepsilon}}) = O(2^{|x|^2})$-time algorithm can break the security of the primitives.
		
		In continuance to the above, by the exact same reductions from the proof of Proposition \ref{prop:soundness} with a single change, we have
		$$
		\Hyb_0 \approx_{c}
		\Hyb_1 \approx_{c}
		\Hyb_2 \approx_{c}
		\Hyb_3 \approx_{c}
		\Hyb_4 \approx_{c}
		\Hyb_5 \approx_{c}
		\Hyb_6 \approx_{c}
		\Hyb_7 \enspace .
		$$
		The single change that we refer to is the check that the reduction makes when getting the final output of the prover.
		In the proof of Proposition \ref{prop:soundness}, the final output of $\zkmP$ is a false proof $\ket{\pi^*}$ for a specific and pre-chosen $x$, while in our case (the adaptive case) it is a pair $(x, \ket{\pi^*})$ for an adaptively-chosen $x$.
		Instead of checking only the verdict of $\zkV$, which can be done in polynomial time, the reduction in our case will also check that $x \in \langno$, which can be done in time $2^{O(|x|)}$.
		This implies that our security reductions take $2^{O(|x|)}$ time to execute, but they break primitives with security against $O(2^{|x|^2})$-time algorithms, which constitutes the needed contradiction.
		Finally, the algorithm $\qsigmaP^*$ that uses $\zkmP$ in the process $\Hyb_7$ in order to break the soundness of the quantum sigma protocol is exactly the same as before, and our proof is finished.
	\end{proof}

	As mentioned before, by the fact that the security of the cryptographic ingridients in our protocol can be based on the hardness of LWE and the security reductions for the primitives are polynomial-time, we get the following corollary.
	
	\begin{corollary}
		Assume there is a constant $\varepsilon \in (0, 1)$ such that LWE is hard for $O(2^{n^{\varepsilon}})$-time quantum algorithms (for LWE secret of $n$ bits).
		Then, for every quantum polynomial-size prover $\zkmP = \{ \zkmP_{\secp} , \rho_{\secp} \}_{\secp \in \Nat}$ there is a negligible function $\mu(\cdot)$ such that for every security parameter $\secp \in \Nat$,
		$$
		\Pr
		\Big[
		(x \in \lang_{no}) \land \big( 1 = \zkV((\crs, \ek), (\prfk, \fhek), x, \ket{\pi^*}) \big)
		\Big] \leq \mu(\secp) \enspace ,
		$$
		where the probability is above the following experiment:
		$$
		(\crs, \ek) \gets \zkSet(1^{\secp}), \enspace 
		\big( (\ciph_{\zkV}, \ciph_{r_\zkV}, \pi_\zkV), (\prfk, \fhek) \big) \gets \zkVSet(\crs, \ek),
		$$
		$$
		(x, \ket{\pi^*}) \gets \zkmP_\secp\big( \rho_\secp, (\crs, \ek), (\ciph_{\zkV}, \ciph_{r_\zkV}, \pi_\zkV) \big)^{\zkV((\crs, \ek), (\prfk, \fhek), \cdot, \cdot)} \enspace .
		$$
		\end{corollary}

	\subsection{Zero Knowledge}
	We show that the protocol is multi-theorem adaptive computational zero-knowledge\footnote{It would have been enough to show that the protocol is \emph{single-theorem} adaptive computational zero-knowledge, and then by the single-to-multi-theorem compiler for NIZKs of \cite{feige1999multiple} get a MDV-NICZK argument with adaptive \emph{multi-theorem} security, but for the sake of completeness, because our construction can be shown to be multi-theorem zero-knowledge without the FLS compilation and because it does not change the main ideas in the proof, we prove the multi-theorem case directly.}, which holds even when the trusted setup samples only a common uniformly random string, and an adversarial polynomial-time (quantum) verifier samples its public verification key maliciously.
	
	We next describe the simulator and then prove that the view that it generates is indistinguishable from the real one, against adaptive distinguishers that choose the statement to be proven only after seeing the common random string.
	
	\paragraph{$\zkSim(1^\secp):$}
	\begin{enumerate}
		\item {\bf CRS Simulation:}
		Given a security parameter $\secp$, the first simulator output is the simulation of the CRS for the NP NIZK protocol and swapping $\ek$ with a public key for the PKE scheme, that is, $\zkSim$ samples:
		$$
		(\crs, \td) \gets \nizkSim(1^\secp) \enspace , 
		(\ek, \sk) \gets \pkeGen(1^\secp) \enspace ,
		$$
		outputs $(\crs, \ek)$ as the simulated CRS and $(\td, \sk)$ as the simulator trapdoor.
		
		\item {\bf Proof Simulation:}
		Given the trapdoor $(\td, \sk)$, a (possibly malicious) public verification key $\pvk = (\ciph_{\zkV}, \ciph_{r_\zkV}, \pi_{\zkV})$ and a yes-instance $x \in \langyes$, the simulator does the following:
		\begin{enumerate}
			\item
			$\zkSim$ checks that $\pi_\zkV$ is a valid proof for the tuple $(\ciph_{\zkV}, \ciph_{r_\zkV}, \ek)$ and also actually verifies some of the statement itself: It decrypts $r_\zkV = \pkeDec_{\sk}(\ciph_{r_\zkV})$ and checks that $\ciph_{\zkV}$ is obtained by running $\prfGen, \fhek \gets \fheGen, \fheEnc_{\fhek}$ with randomness $r_\zkV$.
			If the check is not accepted, $\zkSim$ returns $\bot$.
			
			\item
			$\zkSim$ derives $\prfk$ from $r_\zkV$, computes $\beta_x = \prfF_{\prfk}(x)$ and then executes $(\alpha, \gamma) \gets \qsigmaS(x, \beta_x)$.
			
			\item
			$\zkSim$ performs a circuit-private homomorphic evaluation $\evciph_\zkP \gets \fheEval(C_{\gamma}, \ciph_\zkV)$, where $C_\gamma$ is the circuit that always outputs $\gamma$.
			
			\item
			$\zkSim$ encrypts $\ciph_{r_{\Xi}} \gets \pkeEnc_{\ek}(0^{\ell})$, where $\ell$ is the length of the randomness for the prover in the quantum sigma protocol.
			
			\item
			Finally, $\zkSim$ simulates the non-interactive zero-knowledge proof $\pi_\zkP$, by executing $\pi_{\zkP} \gets \nizkSim(\td, (\evciph_{\zkP}, \ciph_{r_\Xi}, \ek))$.
		\end{enumerate}
		$\zkSim$ outputs $(\alpha, \evciph_{\zkP}, \ciph_{r_{\Xi}}, \pi_{\zkP})$.
	\end{enumerate}
	
	We now prove that the simulated proofs that the simulator generates are computationally indistinguishable from the real proofs that the prover generates.
	
	\begin{proposition} [The Protocol is Multi-theorem Adaptive Computational Zero-knowledge] 
		For every quantum polynomial-size distinguisher $\Disting = \{ \Disting_\secp , \rho_\secp \}_{\secp \in \Nat}$ there is a negligible function $\mu(\cdot)$ such that for every security parameter $\secp \in \Nat$,
		$$
		\abs{P_{\secp, \real} - P_{\secp, \simulated}} \leq \mu(\secp) \enspace ,
		$$
		where,
		$$
		P_{\secp, \real} :=
		\Pr_{(\crs, \ek) \gets \zkSet(1^\secp)}
		\Big[
		\Disting_{\secp}(\rho_{\secp}, (\crs, \ek))^{\zkP((\crs, \ek), \cdot, \cdot, \cdot)} = 1
		\Big] \enspace ,
		$$
		$$
		P_{\secp, \simulated} :=
		\Pr_{((\crs, \ek), (\td, \sk)) \gets \zkSim(1^\secp)}
		\Big[
		\Disting_{\secp}(\rho_{\secp}, (\crs, \ek))^{\zkSim((\td, \sk), \cdot, \cdot)} = 1
		\Big] \enspace ,
		$$
		where in every query that $\Disting$ makes to the oracle, it sends a triplet $(\mpvk, x, \qwit^{\otimes k(\secp)})$ such that $\mpvk$ can be arbitrary, $x \in \langyes\cap \{0, 1\}^\secp$ and $\qwit \in \rel_{\lang}(x)$.
	\end{proposition}
	
	\begin{proof}
		Let $\Disting = \{ \Disting_\secp, \rho_\secp \}_{\secp \in \Nat}$ a polynomial-size quantum distinguisher.
		We prove zero knowledge by a hybrid argument, that is, we consider a series of computationally-indistinguishable hybrid processes with 1-bit outputs, starting from the output of $\Disting$ when getting real proofs, until we get to the output of $\Disting$ when getting simulated proofs.
		We define the following processes.
		\begin{itemize}
			\item $\Hyb_0:$
			The output of $\Disting$ when getting honestly-generated proofs, that is, it gets the CRS from $(\crs, \ek) \gets \zkSet(1^\secp)$ and the proofs from $\zkmP((\crs, \ek), \cdot, \cdot, \cdot)$, as described in the experiment of $P_{\real}$.
			
			\item $\Hyb_1:$
			This hybrid process is identical to $\Hyb_0$, with the exception that $\ek$ is sampled as a public key for the PKE scheme $(\ek, \sk) \gets \pkeGen(1^\secp)$, rather than as a random string of the same length.
			To move to this hybrid we will use the fact that the public keys of the PKE scheme are pseudorandom.
			
			\item $\Hyb_2:$
			This hybrid process is identical to $\Hyb_1$, with the exception that the prover adds another validity check, over the one checking the validity of the proof $\pi_{\zkV}$: It decrypts $r_\zkV = \pkeDec_{\sk}(\ciph_{r_\zkV})$ and checks that $\ciph_{\zkV}$ is obtained by running $\prfGen, \fhek \gets \fheGen, \fheEnc_{\fhek}$ with randomness $r_\zkV$.
			To move to this hybrid we will use the adaptive soundness of the NP NIZK.
			
			\item $\Hyb_3:$
			This hybrid process is identical to $\Hyb_2$, with the exception that we simulate the NP NIZK proofs, that is, (1) when sampling the NP NIZK common random string $\crs$ from the total CRS $(\crs, \ek)$, we sample a simulated CRS $(\crs, \td) \gets \nizkSim(1^\secp)$ instead of $\crs \gets \nizkSet(1^\secp)$, and (2) every time we compute an NP NIZK proof $\pi_\zkP$ as part of the QMA NIZK proof $\ket{\pi}$, we use the NP NIZK simulator $\pi_\zkP \gets \nizkSim(\td, (\evciph_{\zkP}, \ciph_{r_{\Xi}}, \ek))$ rather than $\pi_\zkP \gets \nizkP(\crs, (\evciph_{\zkP}, \ciph_{r_{\Xi}}, \ek))$ (where we execute $\nizkP$ along with a witness for the statement).
			To move to this hybrid we will use the adaptive zero knowledge property of the NP NIZK.
			
			\item $\Hyb_4:$
			This hybrid process is identical to $\Hyb_3$, with the exception that $\ciph_{r_{\Xi}}$ is an encryption of zeros rather than the randomness for the circuit $C_{x, r_\Xi}$, which is homomorphically evaluated.
			To move to this hybrid we will use the security of the PKE scheme.
			
			\item $\Hyb_5:$
			This hybrid process is identical to $\Hyb_4$, with the exception that when computing the evaluated ciphertext $\evciph_{\zkP}$, instead of homomorphically evaluating the circuit $C_{x, r_\Xi}$, we compute $C_{x, r_\Xi}$ in the clear and inject the result by circuit-private evaluation.
			More precisely, the prover does the following: First, it regularly computes $\alpha = \qsigmaP(\ket{w}^{\otimes k(\secp)} ; r_\Xi)$, for randomness $r_\Xi$.
			It derives $\prfk$ from the decrypted randomness $r_\zkV$, computes $\beta_x = \prfF_{\prfk}(x)$, $\gamma = \qsigmaP_3(\beta_x, r_\Xi)$, and then $\evciph_{\zkP} \gets \fheEval(C_{\gamma}, \ciph_{\zkV})$, where $C_{\gamma}$ is the circuit that always outputs $\gamma$.
			To move to this hybrid we will use the circuit-privacy property of the FHE's evaluation algorithm.
			
			\item $\Hyb_6:$
			This hybrid process is identical to $\Hyb_5$, with the exception that when computing $(\alpha, \gamma)$ we use the quantum sigma protocol (special zero-knowledge) simulator, that is, the prover first computes $\beta_x$ (from $\prfk$ which is derived from $r_\zkV$) and then computes $(\alpha, \gamma) \gets \qsigmaS(x, \beta_x)$ and as before, $\alpha$ is sent in the clear and $\gamma$ is sent through homomorphically evaluating the circuit $C_{\gamma}$ on $\ciph_\zkV$.
			To move to this hybrid we will use the special zero knowledge property of the quantum sigma protocol.
			Note that the actions of the prover in this hybrid process are exactly the ones of the QMA NIZK simulator $\zkSim$ and thus $\Hyb_6$ is exactly the process described in the experiment of $P_{\simulated}$.
		\end{itemize}
		We now claim that the outputs of each two consecutive hybrids are computationally indistinguishable, which will finish our proof.
		
		\begin{itemize}
			\item $\Hyb_0 \approx_{s} \Hyb_1:$
			Follows readily from the pseudorandomness property of the public keys generated by $\pkeGen(1^\secp)$.
			
			\item $\Hyb_1 \approx_{s} \Hyb_2:$
			Follows from the adaptive soundness of the NIZK protocol for NP and the statistical correctness of the PKE scheme.
			We explain in more detail: First, note that whenever the NP statement that $\Disting$ proves in $\pi_\zkV$ is correct and the decryption of the PKE is correct, then the output distribution of the proof oracle is identical between the two hybrid processes, as the additional check that is made in $\Hyb_2$ passes successfully. Also note that whenever the proof $\pi_{\zkV}$ is invalid, then both processes output $\bot$ and are identical.
			It follows that the only times that the output distributions of the proof oracles are not identical is whenever there is an error in the decryption of the PKE, or the proof $\pi_{\zkV}$ checks successfully but the statement is false i.e. whenever $\Disting$ breaks the adaptive soundness of the NP NIZK protocol.
			Since both of the above happen with at most negligible probability, it follows that only with negligible probability the outputs of $\Hyb_1$ and $\Hyb_2$ can be distinguished, and the statistical closeness between them follows.
			
			\item $\Hyb_2 \approx_{c} \Hyb_3:$
			Follows readily from the adaptive zero-knowledge property of the NP NIZK protocol.
			
			\item $\Hyb_3 \approx_{c} \Hyb_4:$
			Follows from the security of the PKE scheme.
			Specifically, the encrypted randomness $r_\Xi$ for every query is simply a random string (independent of all other operations in the process) and thus all of these random strings can be chosen at the beginning of the execution of the process, and thus we fix by an averaging argument the strings $r^1_\Xi, \cdots, r^q_\Xi$ that maximize the distinguishability of $\Disting$, where the $q$ is the (polynomial) number of queries that $\Disting$ makes to the proof oracle.
			It then follows that if $\Disting$ distinguishes between $\Hyb_3$ and $\Hyb_4$ then it distinguishes between encryptions of $r^1_\Xi, \cdots, r^q_\Xi$ and encryptions of zeros, and since the single-message security of public-key encryption schemes implies many-message security the indistinguishability $\Hyb_3 \approx_{c} \Hyb_4$ follow.
			
			\item $\Hyb_4 \approx_{s} \Hyb_5:$
			Follows by a hybrid argument, by the circuit-privacy property of the FHE scheme and from the fact that the prover makes the additional check on the public verification key, which checks that $\ciph_{\zkV}$ is obtained by running $\prfGen, \fhek \gets \fheGen, \fheEnc_{\fhek}$ with the extracted randomness $r_\zkV$.
			More precisely, let $q$ be the number of queries that $\Disting$ makes to the proof oracle, and for $i \in \{ 0, 1, \cdots, q \}$ we define $\Hyb^{i}_4$ as the process that performs the homomorphic evaluation of $C_{x, r_\Xi}$ (rather than computing it in the clear and then injecting the result, as done in $\Hyb_5$) starting from query number $i+1$ that $\Disting$ makes, thus $\Hyb^{0}_4 = \Hyb_4$, $\Hyb^{q}_4 = \Hyb_5$.
			
			If $\Hyb_4$ and $\Hyb_5$ are distinguishable then for some $i \in \{ 0, 1, \cdots, q-1 \}$, $\Hyb^{i}_4$ and $\Hyb^{i+1}_4$ are distinguishable. We fix by an averaging argument a snapshot of the execution until after the point that $\Disting$ sends the $(i+1)$-th query to the proof oracle.
			If the check that the prover makes in the beginning, which includes both checking the validity of the NP proof $\pi_{\zkV}$ and also checking the validity of creating $\ciph_{\zkV}$ from the extracted randomness $r_\zkV$, fails, then the hybrid processes are the same as the answer of the proof oracle will be $\bot$.
			In case the check is successful, it follows that the outputs of the circuits $C_{x, r_\Xi}$ and $C_\gamma$ on the input $\prfk$ (which is encrypted inside $\ciph_{\zkV}$) are the same, and thus it follows that the distinguisher between the hybrids $\Hyb^{i}_4$ and $\Hyb^{i+1}_4$ can be used to break the (even statistical) circuit privacy of the FHE evaluation.
			
			\item $\Hyb_5 \approx_{c} \Hyb_6:$
			The proof is very similar to the proof for the indistinguishability $\Hyb_4 \approx_{c} \Hyb_5$, as the indistinguishability follows by a hybrid argument and from the special zero knowledge property of the quantum sigma protocol.
			More precisely, for $i \in \{ 0, 1, \cdots, q \}$ we define $\Hyb^{i}_5$ as the process that uses $\qsigmaP$ (and the polynomially-many copies of the quantum witness) in order to generate $(\alpha, \gamma)$ (rather than computing it using the simulator) starting from query number $i+1$ that $\Disting$ makes, thus $\Hyb^{0}_5 = \Hyb_5$, $\Hyb^{q}_5 = \Hyb_6$.
			
			If $\Hyb_5$ and $\Hyb_6$ are distinguishable then for some $i \in \{ 0, 1, \cdots, q-1 \}$, $\Hyb^{i}_5$ and $\Hyb^{i+1}_5$ are distinguishable. We fix by an averaging argument a snapshot of the execution until after the point that $\Disting$ sends the $(i+1)$-th query to the proof oracle, this in particular fixes the yes instance $x \in \langyes$, the quantum witness $\qwit$ and the pseudorandomness $\beta_x$.
			It follows that the distinguisher between the hybrids $\Hyb^{i}_5$ and $\Hyb^{i+1}_5$ can be used to tell the difference between a tuple $(\alpha, \gamma)$ that was generated by $\qsigmaP$ and a tuple that was generated by $\qsigmaS$, in contradiction the special zero knowledge property of the protocol $(\qsigmaP, \qsigmaV)$.
		\end{itemize}
		
	\end{proof}

\bibliographystyle{alpha}
\bibliography{Bibliography}	

\end{document}